\documentclass[a4paper,11pt]{article}

\usepackage[left=3cm, right=3cm, top=2cm]{geometry}
\usepackage{amssymb}
\usepackage{latexsym}
\usepackage{amsmath}
\usepackage{amsthm}
\usepackage{bm}
\usepackage{graphicx}
\usepackage{hyperref}
\usepackage[inline]{enumitem}
\usepackage{enumerate}
\usepackage{algorithm2e}
\usepackage[round]{natbib}
\usepackage{ogonek}

\newcommand{\Varphi}{\mathit{\Phi}}

\newcommand{\hdr}{\mathsf{H}}
\newcommand{\skr}{\mathsf{S}}
\newcommand{\spayoff}{\Pi^{\skr}}
\newcommand{\hpayoff}{\Pi^{\hdr}}

\newcommand{\neighb}[2]{N_{#2}(#1)}
\newcommand{\eneighb}[2]{N\!E_{#2}(#1)}

\newcommand{\dist}[3]{d_{#1}(#2,#3)}
\newcommand{\entrance}[3]{en_{#2}(#1,#3)}
\newcommand{\exit}[4]{ex_{#2}(#1,#3,#4)}

\newcommand{\pos}{\mathit{pos}}

\newcommand{\seq}{\mathit{seq}}
\newcommand{\act}{\mathit{act}}

\newcommand{\dfs}{\mathsf{DFS}}
\newcommand{\adjdfs}{\mathsf{aDFS}}

\newcommand{\Ex}{\mathbf{E}}

\newcommand{\cycles}{\mathcal{Z}}

\newcommand\xqed{%
  \leavevmode\unskip\penalty9999 \hbox{}\nobreak\hfill
  \quad\hbox{$\triangle$}}

\newtheorem{theorem}{Theorem}
\newtheorem{proposition}{Proposition}

\newtheorem{lemma}{Lemma}

\newtheorem{corollary}{Corollary}

\theoremstyle{remark}
\newtheorem{example}{Example}

\begin{document}
\title{Strategic hiding and exploration in networks\footnote{Marcin Dziubi\'{n}ski's work was supported by the Polish National Science Centre through grant 2018/29/B/ST6/00174.} }

\author{Francis Bloch\footnote{Universit\'{e} Paris 1 and Paris School of Economics, 48
Boulevard Jourdan, 75014 Paris, France.
\texttt{francis.bloch@univ-paris1.fr}},   Bhaskar Dutta\footnote{University of Warwick and Ashoka University, CV4 7AL, Coventry, UK \texttt{b.dutta@warwick.ac.uk}} and Marcin Dziubi\'{n}ski\footnote{Institute of Informatics, University of Warsaw, Banacha 2, 02-097, Warsaw, Poland   \texttt{m.dziubinski@mimuw.edu.pl}}}

\date{}

\maketitle

\abstract{We propose and study a model of strategic network design and exploration where the hider, subject to a budget constraint restricting the number of links, chooses a connected network and the location of an object. Meanwhile, the seeker, not observing the network and the location of the object, chooses a network exploration strategy starting at a fixed node in the network. The network exploration follows the expanding search paradigm of~\cite{AlpernLidbetter2013}.
We obtain a Nash equilibrium and characterize equilibrium payoffs in the case of linking budget allowing for trees only. We also give an upper bound on the expected number of steps needed to find the hider for the case where the linking budget allows for at most one cycle in the network.}

%\keywords{Hide and seek, Unknown network, Expanding search, Network design}

\maketitle 

\section{Introduction}
\label{sec:intro}

Operation of modern businesses and institutions, such as military agencies, hospitals, or universities involves storing and giving access to important and, often times, sensitive data. This access exposes the data to the threat of malicious attacks aiming at stealing or damaging the data. Attacks such as data theft or ransomware attacks (attacks in which the attacker encrypts the crucial data of a company or an institution making it unusable until a ransom is paid) constitute a large fraction of cyber attacks and bear a significant cost to economy. According to European Parliament News~\cite{eunews2023}, ``in 2022 ransomware attacks continued to be one of the main cyberthreats...'' and ``...it is estimated that in 2021 global ransomware reached \texteuro 18 billion worth of damages'', ``... they are also getting more complex''.

A sophisticated attack tactic that is increasingly being used in data breach attacks such as ransomware, data exfiltration, and espionage is \emph{network lateral movement}~\cite{croudstrike2023,cloudflare,sentinelone}. Using network lateral movement, the attacker spreads from an entry point through the network. The objective is to explore the network, learn the network topology, steal credentials and move through the network until valuable assets are found. 
Hence, executing the attack, the attacker searches through the network in an expanding fashion. Based on the knowledge of the network discovered so far, she chooses the next node to move to. Move to subsequent nodes exposes the attacker to the risk of being discovered (as it takes time and also requires the attacker to perform risky actions such as compromising subsequent credentials in the network). Therefore, to make these attacks harder to execute, apart from following good security practices, the network owners can use deception such as decoys to make the attacker easier to discover before the target of attack is reached. In particular, the choice of the network topology affects the speed with which the attacker is able to explore the network. This defence strategy comes at a cost, as larger networks consisting of many nodes are harder to maintain and long sequences of security procedures required to access the data stored in the network require additional time and, possibly, effort from the users.

Another example of situations where threats to network security involve network exploration and where network design can be used to complicate and slow down exploration are situations of infiltration of covert networks, such as criminal, terrorist, or underground movements networks. In these situations, the authorities exert surveillance over the known members of the organization, to learn their contacts and gradually discover new individuals within the organization. The objective of the authorities is to discover the key actors in the organization (e.g. the head of the organization). The objective of the network organizer is to slow down the authorities in order to detect surveillance on time and flee. Again, larger networks are more costly to maintain and longer ``chains of command'' affect negatively the efficiency such covert organizations.

\subsection{Our contribution}
Motivated by these considerations, we propose and study a stylized game theoretic model of network design and exploration where the hider designs a network and chooses a hiding place while the seeker (who is the attacker), aware only of the access point to the network, designs a network exploration strategy that proceeds through the network, gradually discovering new nodes, until the hiding place is found. The model of network exploration follows the expanding search paradigm of~\cite{AlpernLidbetter2013}. The payoff of the hider balances two objectives: he wants to keep the hiding place as close as possible to the entry point, but also wants to maximize the expected number of steps needed by the seeker to find the target. The objective of the seeker is to minimize the expected number of steps to reach the target. 
We find an equilibrium of this game in the case where the linking budget of the hider allows for constructing trees only. In this equilibrium, the seeker uses a randomized depth first search strategy (DFS) while the hider chooses a tree that consists of a line of optimal length, starting at the entry point, and ending with a star of the remaining nodes. We also show that all Nash equilibria are payoff equivalent for the hider and that the equilibrium strategy that we find guarantees at least the hider's equilibrium payoff against any strategy of the seeker. Under an additional assumption of a unique optimal hiding distance from the source for the hider, we show that Nash equilibria are payoff equivalent to the seeker. Moreover the DFS strategy is a best response to any equilibrium strategy of the hider.

In the remaining part of the paper we consider a case where the linking budget allows the hider to construct networks with a cycle. This possibility makes the search problem considerably more complicated, even if the seeker knows the maximal distance from the source at which the hider hides. In particular, the DFS strategy is not as effective as in the case of tree networks. We show that the seeker has a seeking strategy that allows her to find the hider in the expected number of steps that is close to but higher than that guaranteed by the DFS on trees. The strategy is a combination of the DFS, a variant of a bounded DFS, and a seeking strategy that adjusts the DFS after the cycle is discovered in the network. Our findings suggest, in particular, that using cycles as decoys can be an effective method to increase the expected number of steps the seeker needs to find the hider.

\subsection{Related literature}
\label{sec:literature}
The literature on search in networks dates back to the 19th century works of~\cite{Lucas1882} and~\cite{Tarry1895}, who described a search strategy (nowadays known as Tarry's traversal algorithm) that allows for finding a fixed node in an unknown network, starting from any node without visiting links more than twice. A special case of this search strategy is the well known depth first search algorithm~\cite{Even1979}.

Game theoretic analysis of strategic hiding and searching in networks is a part of a large literature on search games (see~\cite{AlpernGal2003,AlpernFokkink2013}). Within this literature, the games that are the closest related to the model studied in this paper are network games with an immobile hider and a fixed starting point of the seeker~\cite{Gal1979,AlpernLidbetter2013}. In these games, the hider chooses a fixed location in a network and aims to maximise the expected time until being captured by the seeker. The seeker chooses a search strategy starting at the given node with the objective of find the hider in minimal time. The games are zero-sum and the focus of the literature is on characterization of optimal seeking and hiding strategies as well as the value of the game~\cite{Gal1979,ReijniersePotters1993,Pavlovic1993,Reijnierse1995,Gal2001,AlpernLidbetter2013}.

Most of the literature in this area considers scenarios where the network is exogenous and known to both players. 
A search model most commonly studied in this literature is the \emph{pathwise search} model~\cite{Gal1979,ReijniersePotters1993,Pavlovic1993,Reijnierse1995,Gal2001}, in which moving from one point in the network to another requires the seeker to traverse all the links on a path between the two points. Links are weighted and the weight of a link represents the time needed to fully traverse the link.
In the case when the hider can hide in nodes as well as at any points within the links, an optimal search strategy must traverse all the links in the network and minimize the total weight of traversed links when doing so. A solution to this problem is a Chinese postman tour~\cite{EdmondsJohnson1973}, which visits each link at most twice and ends at the starting point. As was shown in~\cite{Reijnierse1995} and~\cite{Gal2001}, a mixed strategy which chooses a Chinese postman tour with probability $1/2$ in each direction (called randomized Chinese postman tour) is an optimal search strategy on weakly Eulerian networks.\footnote{
A network is Eulerian when it contains an Euler cycle (which is equivalent to the network being connected and each node having an even degree).
A network is weakly Eulerian when it is obtained by replacing a subset of nodes in a tree with Eulerian graphs.
}\footnote{A similar result for weakly cyclic graphs was earlier shown by~\cite{ReijniersePotters1993}.} 
On these graphs, the strategy guarantees that the expected total weight of the links traversed until the hider is found is equal to at most half of the length of a Chinese postman tour).
On trees, an optimal hiding strategy is a probability distribution on the leaves of the tree, called equal branch density~\cite{Gal1979,Gal2001}. It guarantees that any depth first search of the tree finds the hider in an expected time equal to half of the 
length of a Chinese postman tour which, on trees, is equal to the total weight of the links. This is also the value of 
the game in this case.
%A Chinese postman tour on a tree is a depth first search, this is also the value of the game.
On Eulerian networks, the optimal hiding strategy is mixing uniformly on the set of all points of the graph. The value of the game is half of the total weight of the links, which is also the minimal possible value across all connected networks.
Finding an optimal search strategy for general networks is known to be NP-hard~\cite{vonStengelWerchner1997}.

The search paradigm that we adopt in this paper the \emph{expanding search} introduced by~\cite{AlpernLidbetter2013}. In this paradigm, moving from a node to an unvisited node requires traversing
a minimal weight link connecting the set of visited nodes and the unvisited node.
\cite{AlpernLidbetter2013} studied search games with an exogenous network that is known to both players. In the case of tree networks, they obtain characterization of optimal hiding and searching strategies. They show that the expected time to capture the hider is equal to half of the total weights of the links plus half of a mean distance from the root of the tree to its leaves, weighted according to the equal branch density. Beyond the tree networks, \cite{AlpernLidbetter2013} characterize optimal hiding and searching strategies on $2$-edge-connected networks (a network is $2$-edge-connected if it remains connected after a removal of any link.)
Using the ear decomposition of such networks due to~\cite{Robbins1939}, the authors define a mixed strategy that allows the seeker to reach the hiding node after traversing half of the links, in expectation.
The study of optimal search strategies under expanding search paradigm was continued by~\cite{AlpernLidbetter2019}, where strategies that approximate an optimal search time within a factor close to $1.2$ are described. In particular, the authors define a block optimal strategy which mixes between two different seeking strategies.
In general, the problem of constructing an optimal expanding search strategy, given a graph and a probability distribution over the nodes representing the hiding strategy, is NP-hard~\cite{AverbakhPereira2012}. In recent works, \cite{HermansLeusMatuschke2021} and~\cite{GriesbachHommelsheimKlimmSchewior2023} study efficient approximation algorithms for this problem. In particular, \cite{GriesbachHommelsheimKlimmSchewior2023} propose an algorithm that obtaine a nearly $2e$ approximation ratio.

The crucial difference between the model of~\cite{AlpernLidbetter2013} and the model studied in this paper is that in the former case the network is observed by the seeker before he chooses the seeking strategy, while in our case it is not. In particular, the search strategies for $2$-edge-connected networks, proposed in~\cite{AlpernLidbetter2013}, as well as the strategies described in~\cite{AlpernLidbetter2019} rely on the knowledge of the network and are not applicable in our case. Lack of knowledge of the network makes this approach impossible. The choice of the subsequent nodes to be visited is based on the fragment of the network gradually discovered by the seeker.
In particular, the seeking strategy, that we propose for networks with at most one cycle, does not require the knowledge of the network. Another important difference is that the seeker, knowing that the hider balances between having the hiding place close and hard to find, at the same time, has to take it into account when designing a good search strategy. As we illustrate with examples, this poses additional challenges when the hider's budget allows him to construct cycles.
Lastly, an important difference is the assumption that time is discrete, the hider hides at nodes only, and the seeker moves between nodes in one unit of time. \cite{AlpernLidbetter2013} consider a much more general model, allowing for discrete as well as continuous time. In particular, allowing for continuous time poses significant challenges in formalizing the problem and the expanding search, even if the network is known to the seeker. Wanting to focus on the challenge of formalizing hiding and seeking in an unknown network, we restrict attention to discrete time only.

Other search problems closely related to the model considered in this paper is the search on an unknown network, studied in~\cite{Anderson1981} and~\cite{GalAnderson1990}. In this problem a network is fixed but is unknown to the seeker and one of the nodes in the network is the target node (the ``exit'' of the maze). The links of the network are weighted. The search model is the pathwise search model. \cite{GalAnderson1990} proposed a search strategy called a fixed permutation search (which is a randomized version of the strategy proposed by~\cite{Tarry1895}). The strategy guarantees finding the target node in expected time which is at most the total weight of the links.
\cite{Anderson1981} proposed a zero sum game similar to the one considered in this paper. In this game the hider chooses the network of the given total weight of the links as well as the locations of the starting node and the target node. His objective is to maximize the total weight of links traversed by the seeker. The seeker chooses the search strategy aiming to minimize the payoff of the hider. 
As argued by~\cite{GalAnderson1990}, the fixed permutation search strategy guarantees that the seeker captures the hider in the time not greater than the total weight of the links. Given this result, the optimal strategy of the hider is immediate. The hider constructs a graph consisting of a single link of the given total weight.
There are two main differences between the model of~\cite{Anderson1981} and the model studied in this paper. Firstly, the model of~\cite{Anderson1981} uses the pathwise search model, while in this paper, we consider the expanding search model. Secondly, the hider does not face any cost of locating the target as far from the start as possible. This makes the hiding problem easy to solve (the hider constructs a line consiting of all the nodes and hides the target at the end of the line) and the focus is on designing search strategies that would be good on any networks.
%\footnote{
%In the model of~\cite{Anderson1981} the set of search strategies as well as the set of strategies of the hider is not fully formally defined.}

Further, but still related to the problem studied in this paper, is the literature on the exploration of an unknown network prior to influence maximisation. Motivated by prohibitive cost of sampling large networks, \cite{WilderImmorlicaRiceTambe2018} propose a network exploration strategy which samples a small subset of nodes in the network and then performs a random walk of a limited number of steps in the network starting at these points. The acquired information about the network is then used to determine the probabilities with which the seeding node will be selected from the initially sampled set of nodes. Using experiments and theoretical analysis based on random networks, the authors show that this approach improves upon other influence maximisation methods that do not involve exploring the network. In a similar vein, \cite{EcklesEsfandiari2022} study the problem of influence maximisation with one or more seeds in an unknown network. They assume the independent cascade model and propose a method that first samples the network by seeding a number of nodes and then exploring the network by following the cascade. They consider the variant where the number of links that they sample along the cascade is unbounded as well as a variant where it is bounded. They compare the effectiveness of the two approaches and obtain results on approximation guarantees for the influence maximisation problem using these approaches.

Lastly, the problem studied in this paper is related to the problems of defence and attack in networks~\cite{DziubinskiGoyalVigier2016}. In this literature, the closest related is work by~\cite{BlochDuttaDziubinski2020} that considers a hide and seek game on a network, where the hiders chooses a network and the location of the leader in the network. Then, observing the network but not observing the location of the leader, the seeker chooses a node aiming at being at distance at most one from the location of the leader. The authors characterize a Nash equilibrium of the game and show that an optimal strategy for the hider is to either construct a cycle, a disconnected network, or a core-periphery network. In the two former cases the node of the leader  is chosen uniformly at random from the set of all nodes in the network while in the latter case it is chosen uniformly at random from the set of periphery nodes.

\section{Preliminaries}
\label{sec:prelim}

\paragraph{Graphs.} Throughout the paper we use the following standard graph theoretic notions and notations. An undirected graph over a set of nodes $V$ is a pair 
$G = \langle V, E\rangle$ where $E \subseteq \binom{V}{2}$ is the set of links in $G$ (and $\binom{V}{2}$ denotes the set of all two element subsets of $V$). 
%The unique graph that can be formed over the empty set of nodes, $\varnothing$, is called the empty graph and is denoted by $\varnothing$. 
Given a graph $G$ we use $V(G)$ to denote the set of nodes of $G$ and $E(G)$ to denote the set of links of $G$.

Given node $v \in V$ we use $\neighb{v}{G}$ to denote the set of all neighbours (the \emph{neighbourhood}) of $v$ in $G$, i.e.
$\neighb{v}{G} = \{u \in V : uv \in E(G) \}$. We also use $\eneighb{v}{G}$ to denote the set of all links between $v$ and its neighbours in $G$, i.e.
$\eneighb{v}{G} = \{uv : uv \in E(G)\}$.
Similarly, given a set of nodes $X \subseteq V$ we use $\neighb{X}{G}$ to denote the \emph{(open) neighbourhood} of $X$ in $G$, i.e. the set of neighbours of all the nodes in $X$ in $G$ excluding the nodes in $X$, $\neighb{X}{G} = \bigcup_{v \in X} \neighb{v}{G} \setminus X$, and we use $\eneighb{X}{G}$ to denote the set of links between $X$ and $\neighb{X}{G}$ in $G$, $\eneighb{X}{G} = \bigcup_{v \in X} \eneighb{v}{G}$.

Given a set of nodes $X \subseteq V$ we use $G[X]$ to denote the \emph{subgraph of $G$ induced by $X$}, i.e. $G$ restricted to the nodes in $X$ and the links between them only. Formally $G[X] = \langle X, E[X]\rangle$ where $E[X] = E(G) \cap \binom{X}{2}$. In addition we use $\overline{G}[X]$ to denote the subgraph of $G$ consisting of all the nodes in $X$ and the nodes in the neighbourhood of $X$, all the links between the nodes in $X$ and all the links from the nodes in $X$ and the nodes in the neighbourhood of $X$. Formally,
$\overline{G}[X] = \langle X \cup N_G(X), \overline{E}[X] \rangle$ where $\overline{E}[X] = E[X] \cup \eneighb{X}{G}$. We call such a graph
a \emph{closed subgraph of $G$ induced by $X$} (c.f. Figure~\ref{fig:closedsubgraph} for an example).

\begin{figure}[h]
\begin{center}
  \includegraphics[scale=0.75]{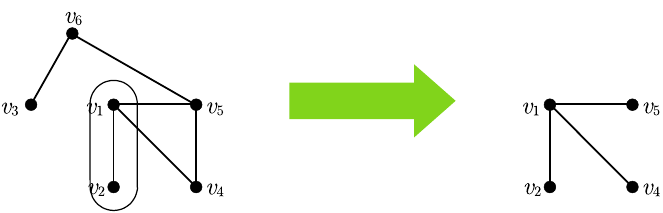}
\end{center}
  \caption{A closed subgraph induced by the set of nodes $\{v_1,v_2\}$.}
  \label{fig:closedsubgraph}
\end{figure}

%Given a graph $G$ and a set of nodes $X \subseteq V(G)$ let $\contr(G,X)$ be a graph obtained from $G$ by contracting the nodes in $X$ into a single node and
%replacing the links from the notes in $X$ to the nodes outside $X$ by links from the new node to the nodes outside $X$. Formally,
%$\contr(G,X) = \langle ( V(G) \setminus X ) \cup \{X\}, E' \rangle$ where, 
%\begin{equation*}
%E' = E(G[V\setminus X]) \cup \{ Xv : \textrm{there exists } u \in X \textrm{ such that } uv \in E(G) \},
%\end{equation*}
%i.e. the nodes in $X$ are removed from $V(G)$ and replaced with a single node, $X$, the links between the nodes in $V\setminus X$ remain unchanged, and every link
%between a node $u \in X$ and a node $v \in V \setminus X$ gets replaced by a link between $X$ and $v$ (c.f. Figure~\ref{fig:contraction} for an example).
%
%\begin{figure}[h]
%\begin{center}
%  \includegraphics{contraction.pdf}
%\end{center}
%  \caption{A graph obtained by contracting nodes in $\{v_1,v_2,v_3\}$.}
%  \label{fig:contraction}
%\end{figure}

Given two nodes $u,v \in V(G)$ a sequence of nodes $v_1,\ldots,v_m$ such that $v_1 = u$, $v_m = v$ and for all $i \in \{2,\ldots,m\}$, $v_{i-1} v_{i} \in E(G)$
is called a \emph{path} between $u$ and $v$ in $G$ and and the length of the path is  the number of links on it, i.e. it is $m-1$. A path is \emph{simple} if every node in $V$ appears at most once in it.
Given nodes $\{s,u\} \subseteq V(G)$ we use
\begin{equation*}
P_G(s,t) = \{v \in V(G) : \textrm{every path from $s$ to $t$ in $G$ contains $v$}\}
\end{equation*}
to denote the set of nodes that belong to every path from $s$ to $t$ in $G$.

Node $v$ is reachable from $u$ in $G$ if there is a path from $u$ to $v$ in $G$. The distance from $u$ to $v$ in $G$ is the length of the shortest path between them in $G$ and is denoted by $\dist{G}{u}{v}$. Graph $G$ is \emph{connected} if for any two nodes $u,v \in V(G)$, $v$ is reachable from $u$ in $G$.

Given a set  of nodes $V$ we use $\mathcal{G}(V)$ to denote the set of all undirected connected graphs that can be formed over $V$. 
We also use $\mathcal{G} = \bigcup_{X \subseteq V} \mathcal{G}(X)$ to denote the set of all undirected connected graphs that can be formed over $V$ or any of its subsets.

\paragraph{Finite sequences.} Given a set $V$ we use $V^{*}$ to denote the set of all finite sequences over the elements of $V$, including the empty sequence, $\varepsilon$. Given a sequence $\bm{z} \in V^{*}$ and an element $v \in V$, $v \cdot \bm{z}$ denotes the sequence obtained by adding $v$ at the beginning of $\bm{z}$ and $\bm{z} \cdot v$ denotes the sequence obtained by adding $v$ at the end of $\bm{z}$. Operator $\cdot$ is the concatenation operator.
Given a sequence of nodes, $\bm{x} \in V^{*}$, and a node $v \in V$ in the sequence, let $\pos(v,\bm{x})$ denote the position at which $v$ appears first in $\bm{x}$. 
%Given a non-empty sequence, $\bm{z}$, $\last(\bm{z})$ is the last element in $\bm{z}$.
Throughout the paper we will abuse the notation and, given a sequence $\bm{z}$, we will also use $\bm{z}$ to denote the set of nodes in $\bm{z}$.

\section{The model}
\label{sec:model}

There are two players, the \emph{hider}, $\hdr$, and the \emph{seeker}, $\skr$, and a set of nodes $V = \{s,v_1,\ldots,v_{n-1}\}$, where $s$ is a distinguished \emph{source node}.
The hider wants to construct an undirected connected graph over $V$, $G \in \mathcal{G}(V)$, and choose a node $h \in V$ where he will hide the treasure. His objective is to have the treasure hidden at a node reachable from the source and his benefit (if the treasure is not found by the seeker), $\Varphi : \mathcal{G}(V) \times V \rightarrow \mathbb{R}_{> 0}$, is a decreasing function of the distance of the hiding node from the source. Let 
\begin{equation*}
\Varphi(G,h) = A(d_G(s,h))
\end{equation*}
where $A : \mathbb{R} \rightarrow \mathbb{R}_{\geq 0}$ is a monotonous non-increasing function.
The hider faces a limited linking budget: the network he chooses must contain at most $b \in \{n-1,\ldots,n(n-1)/2\}$ links.

Knowing the objectives and the budget of the hider, the seeker wants to find the node with the treasure. To do so, he chooses a sequence of nodes that he inspects one by one. Choosing the next node in the sequence,
the seeker observes all the nodes that he visited so far, together with links between them, as well as the set of nodes that are neighbours of the nodes visited so far in the graph chosen by the hider (the frontier nodes). At the beginning the seeker observes node $s$ and its neighbours only.

A \emph{seeking strategy} is a function that maps: 
\begin{itemize}
\item a sequence of nodes starting with $s$ and not containing all nodes in $V$ (e.g. the sequence of nodes visited so far),
\item a graph over the set of nodes in the sequence and the frontier nodes (e.g. the subgraph of the graph chosen by the hider restricted to the nodes visited so far, their neighbours, the links between the nodes visited so far, and the links from these nodes to their neighbours),
\end{itemize}
to a node in the neighbourhood to be visited next. Formally, let $\mathcal{I} \subseteq V^{*} \times \mathcal{G}$ be the set of all pairs
$(\bm{z}, G') \in V^{*} \times \mathcal{G}$ such that
\begin{itemize}
\item $\bm{z} = (z_0,\ldots,z_{k-1})$ with $z_0 = s$, $k < n$, and, for all $j \in \{1,\ldots,k-1\}$, $z_j \in \neighb{\{z_0,\ldots,z_{j-1}\}}{G'}$,
\item $G' = \overline{G'}[\bm{z}]$.
\end{itemize}
A \emph{seeking strategy} is a function $r : \mathcal{I} \rightarrow V$ such that, for all $(\bm{z}, G') \in \mathcal{I}$, $r(\bm{z}, G') \in \neighb{\bm{z}}{G'}$.
The set of all seeking strategies is denoted by $R$.

The hiding and the seeking are modelled as a simultaneous move game where the set of strategies of the hider is $\mathcal{G}(V) \times V$
and the set of strategies of the seeker is $R$.
A strategy profile $((G,h),r) \in (\mathcal{G}(V) \times V) \times R$ determines  a sequence of nodes $\seq(G,r) = v_0,\ldots,v_{n-1}$,
called the \emph{seeking sequence induced by $(G,r)$}, which contains all the nodes in $V$ (each exactly once) and starts with $s$. This is the sequence of nodes visited by the seeker when his seeking strategy $e$ is executed on $G$ until all the nodes are visited. The sequence if formally defined as follows:
\begin{itemize}
\item $v_0 = s$,
\item $\forall i \in \{1,\ldots,n-1\}$, 
$v_{i} = r((v_0,\ldots,v_{i-1}), \overline{G}[\{v_0,\ldots,v_{i-1}\}]).$
\end{itemize}
Payoffs to the players are defined as follows. Fix a strategy profile $((G,h),r) \in (\mathcal{G}(V) \times V) \times R$.
The objective of the seeker is to find node $h$ in as little number of steps as possible. The payoff to the seeker is minus the number of steps from the source
to node $h$,
\begin{equation}
\pi^{\skr}(G,h,r) = -\pos(h,\seq(G,r)).
\end{equation}
The hider enjoys benefit $\Varphi(G,h)$ at every step of the seeker. The payoff to the hider is the benefit of the hider, $\Varphi(G,h)$,
times the number of steps of the seeker from the source to node $h$,
\begin{equation}
\pi^{\hdr}(G,h,r) = \pos(h,\seq(G,r)) \Varphi(G,h).
\end{equation}
A mixed strategy $\eta$ of the hider is a probability distribution on $\mathcal{G}(V) \times V$ and a mixed strategy $\sigma$ of the seeker is a probability distribution on $E$.
The expected payoffs from a given mixed strategy profile, $(\eta,\sigma)$ is defined in the standard way:
\begin{align*}
\spayoff(\eta,\sigma) & = \sum_{(G,h) \in \mathcal{G}(V) \times V}\sum_{r \in R} \eta_{G,h} \sigma_{r} \pi^{\skr}(G,h,r)\\
\hpayoff(\eta,\sigma) & = \sum_{(G,h) \in \mathcal{G}(V) \times V}\sum_{r \in R} \eta_{G,h} \sigma_{r} \pi^{\hdr}(G,h,r).
\end{align*}
We are interested in mixed strategy Nash equilibria of the game defined above.

\section{The analysis}

\subsection{Budget $b = n-1$}
We start the analysis by considering the case of the minimal budget allowing for the construction of a connected network over $n$ nodes: $b = n-1$. With this budget the hider is only able to construct trees.

Consider a seeking strategy called the \emph{randomized depth first search (randomized DFS)}.
The strategy, $\sigma^{\dfs} : \mathcal{I} \rightarrow \Delta(V)$, is defined as follows.\footnote{
Given a finite and non-empty set $X$ we use $\Delta(X)$ to denote the set of all probability distributions on $X$.
}
Given $(\bm{z},G') \in \mathcal{I}$ with $\bm{z} = (z_0,\ldots,z_k)$, let
\begin{equation*}
\act^{\dfs}(\bm{z},G') = z_{\max \{i \in \{0,\ldots,k\} : \neighb{z_i}{G'} \setminus \bm{z} \neq \varnothing \}}
\end{equation*}
and
\begin{equation*}
\sigma^{\dfs}(\bm{z},G')(v) = \begin{cases}
    0, & \textrm{\hspace{-7mm} if $v \in \neighb{\act^{\dfs}(\bm{z},G')}{G'} \setminus \bm{z}$}\\                                           
    \frac{1}{|\neighb{\act^{\dfs}(\bm{z},G')}{G'} \setminus \bm{z}|}, & \textrm{otherwise.}
\end{cases}
\end{equation*}
Given a non-empty sequence $\bm{z}$ and a graph $G'$, $\act^{\dfs}(\bm{z},G')$ is the most recently visited node that has an unvisited neighbour in $G'$.
We call this node an \emph{active node in $\bm{z}$ under $G'$}. If $(\bm{z},G') \in \mathcal{I}$ with $\bm{z} = (z_0,\ldots,z_k)$ then, by the definition of $\mathcal{I}$, $G'$ is connected and $\{z_0,\ldots,z_k\} \subsetneq V$. Therefore $\act^{\dfs}(\bm{z},G')$ is well defined and is unique. 
The strategy $\sigma^{\dfs}$ picks, uniformly at random, one of the unvisited neighbours of the active node determined by $\act^{\dfs}$.

Given a set of $n$ nodes, $V$, such that $s \in V$, a \emph{palm-tree graph} of height $d$ is an undirected graph $\langle V,E\rangle$ such that there exists an ordering of nodes in $V$, $v_0,\ldots,v_{n-1}$ such that:
\begin{itemize}
\item $v_0 = s$,
\item $E = \{v_0 v_1,\ldots,v_{d-2} v_{d-1}\} \cup \{v_{d-1} v_d,\ldots, v_{d-1} v_{n-1}\}$,
\end{itemize}
c.f. Figure~\ref{fig:palmtree} for an illustration. We will refer to the nodes $\{v_d,\ldots,$ $v_{n-1}\}$ as the \emph{crown} of the palm tree and to the nodes $\{v_0,\ldots,v_{d-1}\}$ as the \emph{trunk} of the palm tree.

\begin{figure}[h]
\begin{center}
  \includegraphics[scale=0.75]{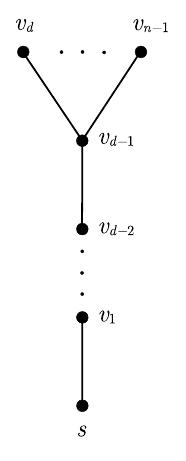}
\end{center}
  \caption{Palm-tree graph.}
  \label{fig:palmtree}
\end{figure}
A mixed hiding strategy where the hider chooses a palm tree of height $d$ and picks one of the nodes of its crown uniformly at random is called a \emph{mixed hiding in a crown of a palm tree of height $d$}.

We now state and prove the equilibrium result for $b = n-1$. 
Let 
\begin{equation*}
D^{\star} = \arg\max_{d \in \{0,\ldots,n-1\}} A(d)\left(\frac{n+d-1}{2}\right).
\end{equation*}

\begin{theorem}
\label{th:budget:1}
If $b = n-1$ then any strategy profile $(\eta,\sigma)$ where $\eta$ is a mixed hiding in a crown of a palm tree of height $d^{\star} \in D^{\star}$ and $\sigma$ is a random DFS is a Nash equilibrium. 

All Nash equilibria yield the same payoff to the hider and strategy $\eta$ guarantees this equilibrium payoff to the hider against any strategy of the seeker. Moreover, strategy $\sigma$ is a best response to any equilibrium strategy of the hider. The equilibrium payoff to the seeker is unique if and only if $|D^{\star}| = 1$.
\end{theorem}

The following two lemmas are key for proving the theorem.~\footnote{The expected number of steps to reach the hiding node $t$ by the randomized DFS strategy, stated in Lemma~\ref{lemma:nosteps}, was characterized before by~\cite{AlpernFokkink2013}. We provide a different proof and we build on and expand this proof when analyzing the search strategy that we propose in Section~\ref{sec:cycles} for the case where the hider's budget allows for the construction of a cycle}

\begin{lemma}
\label{lemma:nosteps}
Let $G = \langle V, E \rangle$ be a tree and let $t \in V$ be a node in $G$. Then the expected number of steps to reach $t$ from any node $s\in V$ in $G$ using randomized DFS strategy starting at $s$ is equal to:
\begin{equation}
\frac{|V| + \dist{G}{s}{t}-m}{2},
\end{equation}
where $m$ is the number of nodes reachable from $s$ in $G$ through a path containing $t$.
\end{lemma}

\begin{proof}
Take any connected graph $G$ that is a tree with starting node $s \in V$ and any $t \in V$. Suppose first that $t$ is a leaf in $G$ (since $G$ is a tree so there is a unique path from $s$ to $t$ in $G$). 
Let $X_t$ be a random variable whose realization is the number of steps made by randomized DFS starting at $s$ until $t$ is reached. 
For any pair of nodes $u,v \in V(G)$, $v \neq u$, let $X_{uv}$ be a random variable such that 
\begin{equation*}
X_{uv} = \begin{cases}
         1, & \textrm{if $u$ is visited before $v$ under $\sigma$}\\
         0, & \textrm{otherwise.}
         \end{cases}
\end{equation*}
Then $X_t = \sum_{v \in V} X_{vt}$ and $\Ex(X_t) = \sum_{v \in V} \Ex(X_{vt})$. Clearly for any $v \in P_G(s,t) \setminus \{t\}$, $\Ex(X_{vt}) = 1$.
We will show that for any $v \in V \setminus P_G(s,t)$, $\Ex(X_{vt}) = 1/2$. Take any such $v$ and let $\{u_1,u_2,u_3\} \subseteq V$ be nodes such that
$\{u_1 u_2,u_1 u_3\} \subseteq E(G)$, $\{u_1,u_2\} \subseteq P_G(s,t)$, and $\{u_1,u_3\} \subseteq P_G(s,v)$ (c.f. Figure~\ref{fig:nosteps} for an illustration).
Since $G$ is a tree, $v$ is not on the path from $s$ to $t$ in $G$ and $t$ is a leaf in $G$, such nodes $u_1$, $u_2$, and $u_3$ exist and are unique.
\begin{figure}[h]
\begin{center}
  \includegraphics[scale=0.75]{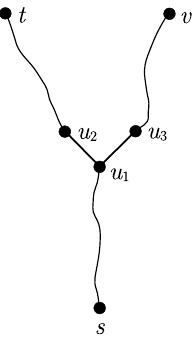}
\end{center}
\caption{Paths from $s$ to $t$ and from $s$ to $v$.}
  \label{fig:nosteps}
\end{figure}

Node $v$ is visited before $t$ if and only if node $u_3$ is visited before $u_2$. Nodes $u_2$ and $u_3$ are visited if an only if $u_1$ is an active node and the probability that $u_3$ is visited before $u_2$ is $1/2$. Hence $\Ex(X_{vt}) = 1/2$.
Thus
\begin{align*}
\Ex(X_t) & = \sum_{v \in V} \Ex(X_{vt}) = |P_G(s,t)|-1 + \frac{|V|-|P_G(s,t)|}{2} \\
         & = \frac{|V|+|P_G(s,t)|-2}{2}
         = \frac{|V|+\dist{G}{s}{t}-1}{2}.
\end{align*}
Since $t$ is a leaf, the number of nodes reachable from $s$ by a path containing $t$ is equal to $1$.

Suppose that $t$ is not a leaf. Then the nodes which are reachable from $s$ via $t$ only will never be visited before $v$ is visited. Suppose that the set of such nodes (excluding $t$) in $G'$ is $U$. Then the expected number of steps to visit $t$ in $G$ by a randomized DFS is equal to the expected number of steps needed in $G[V\setminus U]$, the graph obtained from $G$ by removing all the nodes in $U$. Since $t$ is a leaf in $G[V\setminus U]$, by what we have shown above, the expected number of steps to reach $t$ is equal to $(|V| + \dist{G}{s}{t} - |U| - 1)/2$. Since $m = |U|+1$,
this completes the proof.
\end{proof}

\begin{lemma}
\label{lemma:palmtree}
Let $\beta^d$ be a mixed hiding in a crown of a palm tree of height $d$. Then for any mixed seeking strategy $\alpha$ the expected number of steps to find the hiding place is equal to $(n+d-1)/2$.
\end{lemma}

\begin{proof}
If the hider uses strategy $\beta_d$, any seeking strategy makes $d-1$ steps until it reaches node $v_{d-1}$.
Then, since the hider chooses one of the nodes in $C = \{v_{d^{\star}},\ldots,v_n\}$ as the hiding place uniformly at random, for any sequence of nodes in $C$ the expected number of steps to find the hiding place is equal to $(n-d+1)/2$. Therefore the expected number of steps to reach the hiding place for any seeking strategy is equal to $d-1 + (n-d+1)/2 = (n+d-1)/2$. 
\end{proof}

We are now ready to prove Theorem~\ref{th:budget:1}.

\begin{proof}[Proof of Theorem~\ref{th:budget:1}]
Suppose that $b = n-1$ and let $L(d\mid n) = (n+d-1)/2$.

We first show that $(\eta,\sigma)$ is a Nash equilibrium.
Consider strategy profile $(\eta,\sigma)$, as defined in the theorem. We will show that there is no profitable deviation for any
of the players from this strategy profile. By Lemma~\ref{lemma:palmtree}, if the hider uses strategy $\eta$ then the expected number of steps to reach the hiding place any seeking strategy makes is equal to $L(d^{\star} \mid n)$. Thus, there is no profitable deviation for the seeker against $\eta$.

Consider now any pure strategy $(G',h')$ of the hider. If $G'$ is a palm tree of height $d$ and $h'$ is a node in the crown of $G'$ then the expected payoff to the hider against $\sigma$ is equal to $A(d)L(d \mid n)$. By the definition of $d^{\star}$, $\eta$ yields at least as good of an expected payoff against $\sigma$ as $(G',h')$.
Suppose that $h'$ is a leaf of $G'$ and $G'$ is not a palm tree. Let $d = \dist{G'}{s}{h'}$. By Lemma~\ref{lemma:nosteps}, the expected number of steps to visit $h'$ by $\sigma$ is equal to $L(d)$ and the expected payoff to the hider is equal to $A(d)L(d \mid n)$. By the definition of $d^{\star}$ this payoff is not better
than the payoff the hider gets from $\eta$. Lastly, suppose that $h'$ is not a leaf. Then, by Lemma~\ref{lemma:nosteps}, the expected payoff to the hider from $(G',h')$ is $L(d \mid n-q)$, where $q$ is the number of nodes, other than $h'$, reachable from $s$ in $G'$ by a path containing $h'$. Since $L(d \mid n-q)$ is less than $L(d \mid n)$, by the definition of $d^{\star}$,
$\eta$ yields a higher payoff to the hider than $(G',h')$ against $\sigma$. This shows that all pure strategies of the hider yield at most as high of a payoff against $\sigma$ as $\eta$. Thus this is true for any mixed strategy of the hider and so there is no profitable deviation to the hider from $(\eta,\sigma)$.

Second, we show that all Nash equilibria are payoff equivalent for the hider
To this end we first show, for any Nash equilibrium $(\eta',\sigma')$ and any pair of graph and a hiding node $(G',h')$ in the support of $\eta'$, that the expected number of steps $\sigma'$ makes until $h'$ is reached in $G'$ is equal to $L(d \mid n)$, where $d = \dist{G'}{s}{h'}$.
Assume otherwise and suppose that there exists $(G',h')$ in the support of $\eta'$ such that the expected number of steps $\sigma'$ makes until $h'$ is reached in $G'$ is less than $L(d \mid n)$. 
Let $\beta^d$ be a mixed hiding in a crown of a palm tree of height $d$.
Then, by Lemma~\ref{lemma:palmtree}, the hider would strictly improve her payoff by deviating to the mixed strategy $\sigma''$ under which $(G',h')$ is played with probability $0$ and instead $\beta_d$ is played with probability $\sigma'_{(G',h')}$. This contradicts the assumption that $(\eta',\sigma')$ is a Nash equilibrium. Thus for all $(G',h')$ in the support of $\eta'$, the expected number of steps $\sigma'$ makes until $h'$ is reached in $G'$ is greater or equal to $L(d \mid n)$. Suppose there is $(G',h')$ in the support of $\eta'$ for which it is strictly greater. Then, by Lemma~\ref{lemma:nosteps}, the seeker would strictly benefit from using $\sigma$ instead of $\sigma'$. Again, a contradiction with the assumption that $(\eta',\sigma')$ is a Nash equilibrium. Hence for any $(G',h')$ in the support of $\eta'$, the expected number of steps $\sigma'$ makes until $h'$ is reached in $G'$ is equal to $L(d \mid n)$. 
Now assume, to the contrary, that Nash equilibria are not payoff equivalent to the hider. Then there exists a Nash equilibrium $(\eta',\sigma')$ such that $\hpayoff(\eta',\sigma') \neq \hpayoff(\eta,\sigma)$. If $\hpayoff(\eta',\sigma') < \hpayoff(\eta,\sigma)$ then, by Lemma~\ref{lemma:palmtree} and optimality of $d^{\star}$, the hider would strictly benefit from using $\eta$ instead of $\eta'$ against $\sigma'$. Hence it must be that $\hpayoff(\eta',\sigma') > \hpayoff(\eta,\sigma)$. Then there exists $(G',h')$ in the support of $\eta'$ which yields higher payoff to the hider than $A(d^{\star})L(d^{\star} \mid n)$. As we established above, the expected number of steps $\sigma'$ makes until $h'$ is reached in $G'$ is equal to $L(d \mid n)$, where $d = \dist{G'}{s}{h'}$. Hence the payoff to the hider from $(G',h')$ is $A(d)L(d \mid n)$ and, by optimality of $d^{\star}$ it cannot be strictly greater than $A(d^{\star})L(d^{\star} \mid n)$, a contradiction. Therefore all equilibria must be payoff equivalent to the hider.
Moreover, by Lemma~\ref{lemma:palmtree}, $\eta$ yields the equilibrium payoff to the hider against any seeking strategy.

Third, we show that strategy $\sigma$ is a best response to any equilibrium strategy of the hider. Since, as we showed above, for any Nash equilibrium $(\eta',\sigma')$ and any pair of a graph and a hiding node $(G',h')$ in the support of $\eta'$, that the expected number of steps $\sigma'$ makes until $h'$ is reached in $G'$ is equal to $L(d \mid n)$, where $d = \dist{G'}{s}{h'}$, so $\spayoff(\eta',\sigma') = \spayoff(\eta',\sigma)$. Since $\sigma'$ is a best response to $\eta'$ so is $\sigma$ as well.

Fourth, we show that equilibrium payoff to the seeker is unique of and only if $|D^{\star}| = 1$. If $|D^{\star}| = 1$ that the claim follows because, as we showed above, $\sigma$ is a best response to any equilibrium strategy of the hider and the payoff to the seeker from any equilibrium $(\eta',\sigma')$, $\spayoff(\eta',\sigma') = \hpayoff(\eta',\sigma')/A(d^{\star})$, where $\{d^{\star}\} = D^{\star}$. Since $\hpayoff(\eta',\sigma')$ is unique (as we showe above) and $d^{\star}$ is unique, so is $\spayoff(\eta',\sigma')$. Now, suppose that $|D^{\star}| > 1$. Then there exist $d^{\star} \in D^{\star}$ and $d^{\star\star} \in D^{\star}$ such that $d^{\star} < d^{\star\star}$. As we have shown above, both $(\eta,\sigma)$, with $\sigma = \beta_{d^{\star}}$, and $(\eta',\sigma)$, with $\eta' = \beta_{d^{\star\star}}$, are Nash equilibria of the game. Moreover, $\spayoff(\eta,\sigma) = L(d^{\star}) < L(d^{\star\star}) = \spayoff(\eta',\sigma)$.
\end{proof}

\subsection{Budget $b = n$}
\label{sec:cycles}
With budget $b = n$ the hider can build a network that contains a cycle. This possibility makes the search problem considerably more challenging. 

Since the problem is significantly more challenging than in the case of $b = n-1$, to make progress, we assume a particular form of payoffs from hiding at a given distance, $A$. For a fixed $d > 0$, let
\begin{equation*} 
A(x) = \begin{cases}
       1, & \textrm{if $x \leq d$},\\
       0, & \textrm{otherwise.}
       \end{cases}
\end{equation*}
Thus the hider gets payoff $1$ if the target node is at distance at most $d$ from the source and payoff $0$ otherwise.
Notice that this makes the game zero-sum.
%We propose a seeking strategy that guarantees the expected number of steps to the target to be bounded from above by $(1/2 + 1/16)n$ plus a constants that depends on $d$. The seeking strategy, $\sigma : \mathcal{I} \rightarrow \Delta(V)$ randomizes between three seeking strategies, the DFS strategy, a DFS strategy bounded by $d > 0$, which first explores the nodes at distance at most $d$ from the source and then the remaining nodes, and an adjusted DFS strategy which, after discovering
%a cycle, visits the nodes which are reachable by a single path first and the nodes which are reachable by two paths after them.
%
In particular, the hider
can create networks on which the randomized DFS makes more than $2n/3$ steps to find the target, in expectation. We illustrate this possibility with the following example.

\begin{example}
\label{ex:1}
Consider a network over $n$ nodes consisting of a line of $d + 1$ nodes, starting at node $s$, and a cycle over $n-d$ nodes, including $s$ (c.f. Figure~\ref{fig:example1}). Consider the randomized DFS seeking strategy starting at node $s$. Every node on the line from $s$ to $t$ is visited with probability $1$ until $t$ is visited. In addition, every node in the cycle (except $s$) is visited with probability $2/3$ before $t$ is visited. Hence $2/3(n+d/2-1)$ nodes will be visited, in expectation, until $t$ is visited.

\begin{figure}[h]
\begin{center}
  \includegraphics[scale=0.75]{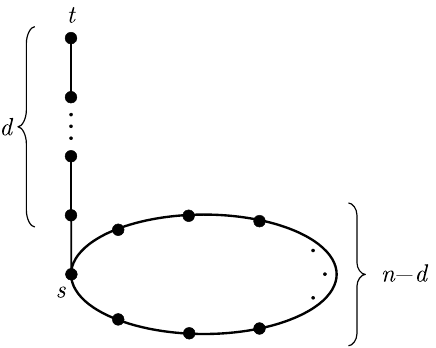}
\end{center}
\caption{A graph over $n$ nodes that requires $2/3(n+d/2-1)$ steps for the randomized DFS to reach node $t$ starting at node $s$.}
  \label{fig:example1}
\end{figure}
\xqed
\end{example}

As illustrated by the example, the hider can put a large fraction of the nodes in the cycle and hide the treasure outside the cycle, which makes the seeker visit the nodes in the cycle with high probability before the hiding node outside the cycle is visited. Since the payoff to the hider is decreasing in the distance from the starting node to the hiding node, it is not beneficial to the hider to hide further than a distance determined by benefit from hiding and the seeking strategy. Hence, the seeker could adjust the seeking strategy so that all the nodes which are at certain distance, $d$, from the starting node are visited before the nodes which are further than $d$ from the starting node. 
This motivates a definition of a seeking strategy called the \emph{$d$-bounded randomized DFS}, denoted by $\dfs_d$,
where $d \geq 1$ is a parameter of the strategy.

Informally, the strategy is described as follows. On a tree, the $\dfs_d$ seeking strategy visits all the nodes at distance at most $d$ from the source first, proceeding analogously to the $\dfs$ seeking strategy. After all such nodes are visited, the seeking strategy visits all the remaining nodes following the $\dfs$ strategy.
On graphs containing a cycle, the description of the strategy is more complicated because after the cycle is discovered (i.e. the subgraph of the explored graph restricted to the set of visited nodes contains a cycle), the distances from the source to the nodes in the cycle need to be updated. Some nodes, where visiting was postponed due to their distance from the source being greater than $d$, may turn out to be at distance at most $d$ from the source when the second path to them (discovered together with discovering the cycle) is discovered. These nodes (and the nodes reachable through them from the sources) are visited before the remaining nodes as soon as the cycle is discovered. Moreover, 
they are visited in the ``first seen first visited'' order rather than ``last seen first visited'' order used by the standard $\dfs$ strategy. 

Before providing a formal definition of the seeking strategy we need to introduce some notions and notation. 
Given a graph $G$, a set of nodes $\{v_1,\ldots,v_m\}$ such that $v_m v_1 \in E(G)$ and, for all $i \in \{2,\ldots,m\}$, $v_{i-1} v_i \in E(G)$ is a \emph{cycle} in $G$. The set of all cycles in $G$ is denoted by $\cycles(G)$.
Given $d \geq 1$, $i \geq 1$, graph $G$, and node $s \in V(G)$, let
\begin{equation*}
R^i_d(G,s) = \{v \in V(G) : \textrm{there is exactly $i$ paths of length $\leq d$ from $s$ to $v$}\}
\end{equation*}
be the set of nodes in network $G$ which are reachable by exactly $i$ paths of length at most $d$ from $s$ in $G$.
Let $R_d(G,s) = \bigcup_{i \geq 1} R^i_d(G,s)$ be the set of all the nodes at distance at most $d$ from $s$ in $G$.

Like in the $\dfs$ seeking strategy, the nodes to be visited next are picked uniformly at random from the set of unvisited neighbours of the active node. The main difference between the $\dfs_d$ and $\dfs$ seeking strategies is in how the active node is selected.
The \emph{$d$-bounded DFS} seeking strategy, $\sigma^\mathrm{\dfs_d} : \mathcal{I} \rightarrow \Delta(V)$, is formally defined as follows. Given an integer bound $d \geq 0$, for any $(\bm{z},G') \in \mathcal{I}$ with $\bm{z} = (z_0,\ldots,z_k)$, $\act^{\dfs_d}(\bm{z},G') = z_j$, where $j$ is selected as follows

\begin{algorithm}%[H]
%\AlgoDontDisplayBlockMarkers
\DontPrintSemicolon
\SetAlgoNoEnd
%\SetAlgoNoLine

    \uIf{$\cycles(G'[\bm{z}]) \neq \varnothing$ and $\neighb{\bm{z}}{G'} \cap R^1_d(G',s) \cap \left(R^2_n(G',s) \setminus R^2_{d+1}\right)  \neq \varnothing$}{
         $j = \max \Big\{i \in \{0,\ldots,k\} : (\neighb{z_i}{G'}\setminus \bm{z}) \cap R^1_d(G',s) \cap\allowbreak \left(R^2_n(G',s) \setminus R^2_{d+1}\right) \neq \varnothing \Big\}$\;
    }
    \uElseIf{$\cycles(G'[\bm{z}]) \neq \varnothing$ and $\neighb{\bm{z}}{G'} \cap R^1_d(G',s) \cap \left(R^2_{d+1}(G',s)\setminus R^2_{d}(G',s)\right)  \neq \varnothing$}{
        $j = \min \Big\{i \in \{0,\ldots,k\} : (\neighb{z_i}{G'}\setminus \bm{z}) \cap R^1_d(G',s) \cap\allowbreak \left(R^2_{d+1}(G',s)\setminus R^2_{d}(G',s)\right) \neq \varnothing \Big\}$\;
    }
    \uElseIf{$\neighb{\bm{z}}{G'}\cap R_d(G',s) \neq \varnothing$}{
        $j = \max \{i \in \{0,\ldots,k\} : (\neighb{z_i}{G'}\setminus \bm{z}) \cap  R_d(G',s) \neq \varnothing \}$\;
    }
    \uElse{
        $j = \max \{i \in \{0,\ldots,k\} : \neighb{z_i}{G'}\setminus \bm{z} \neq \varnothing \}$\;
    }
\end{algorithm}
\noindent and
\begin{align*}
& \sigma^{\dfs}_{d}(\bm{z},G')(v) = {} \\
& \quad \begin{cases}
\frac{1}{|(\neighb{\act^{\dfs}_{d}(\bm{z},G')}{G'} \setminus \bm{z}) \cap R_d(G',s)|},\\ 
& \hspace{-37mm} \textrm{if $v \in \neighb{\act^{\dfs}_{d}(\bm{z},G')}{G'}\setminus \bm{z}) \cap R_d(G',s)$}\\
\frac{1}{|(\neighb{\act^{\dfs}_{d}(\bm{z},G')}{G'} \setminus (\bm{z} \cup R_d(G',s)))|},\\
& \hspace{-37mm}\parbox[t]{7cm}{if $v \in \neighb{\act^{\dfs}_{d}(\bm{z},G')}{G'}\setminus (\bm{z} \cup R_d(G',s))$ and $\neighb{\act^{\dfs}_{d}(\bm{z},G')}{G'}\setminus \bm{z}) \cap R_d(G',s) = \varnothing$,}\\
                                           0, & \hspace{-37mm}\textrm{otherwise.}
                                           \end{cases}
\end{align*}

The first two cases of the active node selection function apply only when the graph over the visited nodes contains a cycle. Notice that the nodes, whose visiting was postponed due the fact that the distance from the source is to too high, and which are visited after the cycle is discovered, are the nodes reachable by exactly one path of length at most $d$ and by exactly two paths of length at most $d+1$. The second case of the active node selection function applies to them.
The first case applies to the nodes which are reachable by exactly one path of length at most $d$ and by exactly two paths of length $d+2$ or more. These are the nodes reachable via the postponed nodes mentioned above and they are visited in the ``last seen first visited order'' of the standard $\dfs$ strategy.

The randomized DFS bounded by $d$ visits $2d$ nodes in expectation until node $t$ is visited starting from node $s$
in the network in Figure~\ref{fig:example1}. However, the hider can still make the seeker visit more than $2n/3$ nodes in expectation by having many nodes reachable by two paths from the sources and at distance $d$ from the source. This is illustrated with the following example.

\begin{example}
\label{ex:2}
Consider a network over $n$ nodes consisting of a line of $d + 1$ nodes, starting at node $s$, and a cycle over $2d-2$ nodes, including $s$, and $n-3d+2$ nodes not in the cycle, reachable by two paths of length $d$ from $s$ (c.f. Figure~\ref{fig:example2}). Similarly to the network in Example~\ref{ex:1}, the randomized DFS seeking strategy starting at node $s$ visits $2/3(n+d/2-1)$ nodes in expectation, until $t$ is visited.
The randomized DFS bounded by $d$ visits $2/3(n+1/2)$ nodes in expectation until node $t$ is visited. This is because every node on the line from $s$ to $t$, is visited with probability $1$, every node at distance $d$ from $s$ reachable by two paths from $s$ is visited with probability $2/3$ before $t$ is visited (and there are $n-3d+5$ such nodes, $3$ in the cycle and $n-3d+2$ outside the cycle). Each of the remaining $2(d-3)$ nodes (each of them in the cycle) is visited with probability $1/2$ before $t$ is visited.

\begin{figure}[h]
\begin{center}
  \includegraphics[scale=0.75]{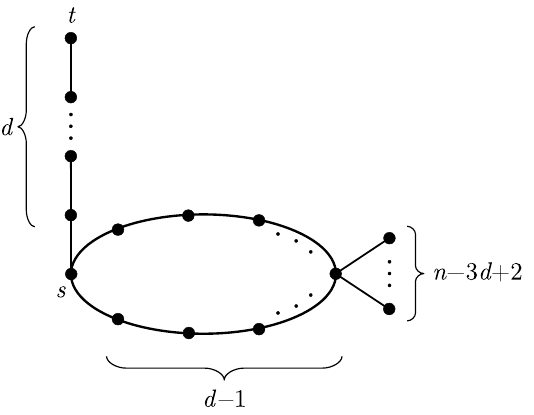}
\end{center}
\caption{A graph over $n$ nodes that requires $2/3(n+d/2-1)$ steps for the randomized DFS and $2/3(n+1/2)$ steps for the randomized DFS bounded by $d$ to reach node $t$ starting at node $s$.}
  \label{fig:example2}
\end{figure}
\xqed
\end{example}

Large expected number of visited nodes until the hiding node is visited results from the fact that $O(n)$ nodes are at distance $d$ from the source and are reachable by two paths from $s$. A randomized DFS seeking strategy (bounded or unbounded) visits these nodes with probability $2/3$ before the node reachable by a single path is visited. To reduce this probability, the DFS can be adjusted so that it postpones visiting nodes reachable by two paths and visits nodes reachable by a single path first. This can be done after the graph over the visited nodes contains a cycle and requires a modification to an unbounded DFS strategy, because the bounded version cannot discover cycles which contain more than $d+1$ nodes (the network used in Example~\ref{ex:2} has a cycle over $2d$ nodes which cannot be discovered by the DFS bounded by $d$).
To define the adjusted DFS seeking strategy we need to introduce some more notions and notation. 

Given a cycle $K \in \cycles(G)$ and a node $s \in V$, $\entrance{G}{K}{s} \in K$ denotes the node of cycle $K$ that is closest to $s$ (we call it the \emph{entrance to the cycle from $s$}). Given $d \geq 1$, $i \geq 1$, graph $G$, and nodes $\{s,u\} \subseteq V(G)$, let
\begin{multline*}
R^i_d(G,s,u) = \{v \in V(G) : \textrm{there is exactly $i$ paths} \\ 
\textrm{of length $\leq d$ from $s$ to $v$ containing $u$}\}
\end{multline*}
be the set of nodes reachable from $s$ by exactly $i$ paths of length at most $d$, each path containing node $u$. In particular, if $G$ contains a cycle $K$, $R^1_n(G,s,\entrance{G}{K}{s})$ is the set of nodes reachable from $s$ by exactly one path that contains the entrance to the cycle from $s$.

The \emph{adjusted DFS} seeking strategy, $\sigma^\mathrm{\adjdfs} : \mathcal{I} \rightarrow \Delta(V)$ is formally defined as follows. Given any 
$(\bm{z},G') \in \mathcal{I}$ with $\bm{z} = (z_0,\ldots,z_k)$, $\act^{\adjdfs}(\bm{z},G') = z_j$ where , where $j$ is selected as follows

\begin{algorithm}%[H]
%\AlgoDontDisplayBlockMarkers
\DontPrintSemicolon
\SetAlgoNoEnd
%\SetAlgoNoLine

    \uIf{$\cycles(G'[\bm{z}]) = \{K\}$ and $\neighb{\bm{z}}{G'} \cap \left(R^1_n(G',s,\entrance{G'}{K}{s})\setminus R^2_n(G',s)\right) \neq \varnothing$}{
         $j = \max \{i \in [k] : (\neighb{z_i}{G'}\setminus \bm{z}) \cap  R^1_n(G',s) \neq \varnothing \}$\;
    }
    \uElseIf{$\cycles(G'[\bm{z}]) = \{K\}$ and $\neighb{\bm{z}}{G'}\cap R^2_n(G',s) \neq \varnothing$}{
        $j = \max \left\{i \in [k] : (\neighb{z_i}{G'}\setminus \bm{z}) \cap R^2_n(G',s) \neq \varnothing \right\}$\;
    }
    \uElse{
        $j = \max \{i \in [k] : \neighb{z_i}{G'} \setminus \bm{z} \neq \varnothing \}$\;
    }
\end{algorithm}
\noindent and
\begin{equation*}
\sigma^{\adjdfs}(\bm{z},G')(v) =
 \begin{cases}
\frac{1}{|\neighb{\act^{\adjdfs}(\bm{z},G')}{G'} \setminus \bm{z}|},\\
 & \hspace{-10mm} \textrm{if $v \in \neighb{\act^{\adjdfs}(\bm{z},G')}{G'}\setminus \bm{z}$}\\
                                           0, & \hspace{-10mm} \textrm{otherwise.}
                                           \end{cases}
\end{equation*}
Like the DFS seeking strategy, strategy $\sigma^{\adjdfs}$ chooses, uniformly at random, an unvisited neighbour of an active node. The active node is a visited
node selected by function $\act^{\adjdfs}$. After the cycle is discovered, i.e. the graph over the visited nodes contains a cycle, the function chooses
the nodes reachable by a single path from $s$ via the entrance to the cycle before the visited nodes that have neighbours reachable by two paths from $s$ (since the graph has at most one cycle, the number of paths to the neighbours of unvisited nodes can be determined after the cycle is discovered).

The seeking strategy $\sigma^{\star}$ that we define below is a convex combination of three seeking strategies: the randomized DFS, depth bounded randomized DFS and adjusted DFS.
Let $\sigma^{\star} : \mathcal{I} \rightarrow \Delta(V)$ be a seeking strategy formally defined as follows. Given any $(\bm{z},G') \in \mathcal{I}$ let
\begin{multline*}
\sigma^{\star}(\bm{z},G')(v) = {} \\
\frac{3}{8}\sigma^{\dfs}(\bm{z},G')(v) + \frac{3}{8}\sigma^{\adjdfs}(\bm{z},G')(v) + \frac{1}{4}\sigma^{\dfs}_d(\bm{z},G')(v).
\end{multline*}
The following proposition (proven in the Appendix) provides a lower bound on seeker's payoff for the case when linking budget of the hider is $n$.

\begin{proposition}
\label{pr:cycle:lowerbound}
For any connected network $G$ and node $h \in V(G)$ at distance at most $d$ from $s$, the expected number of steps to discover the hider using strategy $\sigma^{\star}$ is less than or equal to
\begin{equation*}
\frac{9}{16}n + C(d),
\end{equation*}
where $C(d)$ depends on $d$ only.
\end{proposition}

\section{Conclusions}
We studied a problem of strategic hiding and exploration of a network and provided a formal definition of the game theoretic model of the problem. We characterized a Nash equilibrium of the game when the hider's budget allows for constructing trees only. The equilibrium strategy of the hider that we found guarantees equilibrium payoff against any seeking strategy and all equilibria yield the same payoff to the hider. The strategy of the seeker in the equilibrium that we found is the randomized DFS strategy. We showed that if an optimal hiding distance from the source is unique then all equilibria are payoff equivalent to the seeker and the DFS strategy guarantees the same expected number of steps to find the hider against any equilibrium strategy of the hider.

We also considered the case where the budget of the hider allows for the construction a network that contains a cycle. We showed that the randomized DFS strategy is no longer effective in this case and proposed a seeking strategy that finds the hider, hiding at distance $d$ from the source in the expected number of steps close, but greater, to the expected number of steps guaranteed by the DFS strategy on trees.

There are a number of possible avenues for future research. The most interesting, but most likely a very ambitious one, would be to characterize equilibria (or at least one equilibrium) in the case where the hider does not face any budget constraints. Another interesting question would be about the seeking strategy against networks containing a cycle. Is it possible to construct a seeking strategy that would find the hider hiding at distance $d$ from the source in $n/2 + C(d)$ steps, where $C(d)$ depends on $d$ only? Can the number of steps $9/16 n + C(d)$ provided in Proposition~\ref{pr:cycle:lowerbound} be improved?

%\section{Declarations}
%
%\paragraph{Funding.} Marcin Dziubiński’s work was supported by the Polish National Science Centre through grant
%2018/29/B/ST6/00174.

\bibliographystyle{plainnat}
\bibliography{biblio}

\newpage

\appendix

\section*{Appendix}

\section*{Lower bound on seeker's payoff when $b = n$}
In this section we prove Proposition~\ref{pr:cycle:lowerbound} stating the upper bound on the number of nodes visited by seeking strategy $\tilde{\sigma}$ until the node chosen by the hider is found. 

The key to prove the proposition is characterization of the expected value of random variables $X_{vt}$ across the pairs of distinct nodes $v$ and $t$. The following lemma provides a characterization of the values of $X_{vt}$ under a bounded DFS strategy for the cases when $t$ is at distance at most $d$ from the source.

\begin{lemma}
\label{lemma:nosteps:dfsd}
Let $d > 0$, $G = \langle V, E\rangle$ be a graph containing at most one cycle $K$ such that $K \cap R^2_d(G,s) \neq \varnothing$, and let $t \in V$ be either a leaf or a node in $K$. Let $0 \leq \dist{G}{s}{t} \leq d$. Under bounded DFS seeking strategy $\sigma^{\dfs}_d$, for any node $v \in V\setminus P_G(s,t)$ such that $t\notin P_G(s,v)$,
\begin{itemize}
\item if $t \in R^{1}_n(G,s) \setminus R^1_n(G,s,\entrance{G}{K}{s})$ then
\begin{equation*}
\Ex(X_{vt}) = \begin{cases}
              \frac{1}{2}, & \textrm{if $\dist{G}{s}{v} \leq d$,}\\
              0, & \textrm{if $\dist{G}{s}{v} > d$.}
              \end{cases}
\end{equation*}
\item if $t \in R^1_n(G,s,\entrance{G}{K}{s})$ then
\begin{equation*}
\Ex(X_{vt}) = \begin{cases}
              \frac{2}{3}, & \textrm{if $v\in R^2_d(G,s)$,}\\
              \frac{1}{2}, & \textrm{if $v\in R^1_d(G,s) \cap R^1_n(G,s)$,}\\
              \frac{2}{3}, & \textrm{if $v\in R^1_d(G,s) \cap R^2_n(G,s)$ and $K\subseteq R^2_d(G,s)$,}\\
              \frac{1}{2}, & \textrm{if $v\in R^1_d(G,s) \cap R^2_n(G,s)$ and $K\setminus R^2_d(G,s) \neq \varnothing$,}\\
              \end{cases}
\end{equation*}
\item if $t \in K$ and $v \in K$ then $\Ex(X_{vt}) = 1/2$,
\item if $t \in R^2_n(G,s) \setminus K$ and $v \in K$ then
\begin{equation*}
\Ex(X_{vt}) = \begin{cases}
              1, & \parbox[t]{8cm}{if $t \in R^1_d(G,s) \cap R^1_d(G,s,v)$,}\\
              \frac{3}{4}, & \parbox[t]{8cm}{if $t \notin R^1_d(G,s) \cap R^1_d(G,s,v)$ and $v \in R^2_d(G,s)$,}\\
              \frac{1}{2}, & \parbox[t]{8cm}{if $t \notin R^1_d(G,s) \cap R^1_d(G,s,v)$ and $v \in R^1_d(G,s)$.}\\
              \end{cases}
\end{equation*}

\item if $\{t,v\} \subseteq R^2_n(G,s) \setminus K$ and $K\subseteq R^2_d(G,s)$ then
\begin{align*}
\Ex(X_{vt}) & =  \frac{1}{2} + \frac{1}{8}[v \in R^2_d(G,s)][t \in R^1_d(G,s) \cap R^1_d(G,s,\exit{G}{K}{s}{v})]
\end{align*}

\item if $\{t,v\} \subseteq R^2_n(G,s) \setminus K$ and $K\setminus R^2_d(G,s) \neq \varnothing$ then
\begin{align*}
\Ex(X_{vt}) =  \frac{1}{2} & + \frac{1}{4}[t \in R^1_d(G,s)][v \in R^2_d(G,s)]
\end{align*}

\end{itemize}
\end{lemma}

\begin{proof}
Take any connected graph $G$ with starting node $s \in V$ such that either $\cycles(G) = \varnothing$ or $\cycles(G) = \{K\}$, so that $K$ is the unique cycle in $G$, and $K \cap R^2_d(G,s) \neq \varnothing$, so that at least one node in the cycle is reachable by exactly two paths of length at most $d$. Let $t \in V$ be either a leaf or a node in $K$ such that $0 \leq \dist{G}{s}{t} \leq d$.
%It follows immediately from the definition of $X_{uv}$ that $X_{uv} = 1 - X_{vu}$. Therefore, by linearity of expected value, $\Ex(X_{uv}) = 1 - \Ex(X_{vu})$. 
Let $v \in V(G)$ be any node such that $v \neq t$, $t \notin P_G(s,v)$, and $v\notin P_G(s,t)$. 

Since the seeking strategy visits nodes at distance at most $d$ from $s$ first so $\Ex(X_{vt}) = 0$ for all $v$ with $\dist{G}{s}{v} > d$. For the remaining part of the proof assume that $\dist{G}{s}{v} \leq d$.

Suppose that $t \in R^1_n(G,s) \setminus R^1_n(G,s,\entrance{G}{K}{s})$, i.e. there is a unique path from $s$ to $t$ in $G$ and that path does not contain any node of the cycle. Since $v \in V(G) \setminus P_G(s,t)$ and $t \notin P_G(s,v)$ so $\Ex(X_{tv}) = 1/2$. 
This follows by the same arguments as those used in proof of Lemma~\ref{lemma:nosteps}.

Suppose that $t \in R^1_n(G,s,\entrance{G}{K}{s})$, i.e. there is a unique path from $s$ to $t$ in $G$ and that path contains a node of the cycle. In this case
the path from $s$ to $t$ contains $\entrance{G}{K}{s}$ and no other node of the cycle. 

Suppose that $v \in R^1_n(G,s)$. Then there is exactly one path from $s$ to $v$ in $G$ and, since $\dist{G}{s}{v} \leq d$, this path is of length at most $d$.
In this case $\Ex(X_{vt}) = 1/2$, by analogous arguments to those used in proof of Lemma~\ref{lemma:nosteps}. 

Suppose that $v \in R^2_n(G,s)$. Since $\dist{G}{s}{v} \leq d$ so either $v \in R^1_d(G,s)$ or $v \in R^2_d(G,s)$.
Let $u_1$ be the successor of $\entrance{G}{K}{s}$ on the path from $s$ to $t$, and let $u_2$ and $u_3$ be the two distinct successors of $\entrance{G}{K}{s}$ on the paths from $s$ to $v$ such that $\dist{G}{u_2}{v} \leq \dist{G}{u_3}{v}$. In the case of $v \notin K$, let $u_4$ be the successor $\exit{G}{K}{s}{v}$ on the path from $s$ to $v$ (c.f. Figure~\ref{fig:conf1}).

\begin{figure}[h]
\begin{center}
  \includegraphics{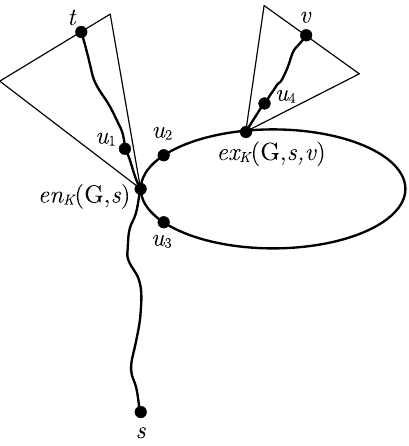}
\end{center}
\caption{Paths from $s$ to $t$ and from $s$ to $v$ in the case of $t \in R^1_n(G,s,\entrance{G}{K}{s})$ and $v \in R^2_n(G,s)$.}
  \label{fig:conf1}
\end{figure}

Suppose that $K \subseteq R^2_d(G,s)$ (i.e. every node in the cycle is reachable by exactly two paths of length at most $d$ from $s$). In this case the set of visited nodes contains all the nodes of the cycle if and only if either $u_2$ or $u_3$ is visited before $u_1$. Thus if $v \in K$ then $v$ is visited before $t$
with probability $2/3$ and so $\Ex(X_{vt}) = 2/3$. Suppose that $v \notin K$. In this case $v$ is visited before $t$ if and only if $u_4$ is visited before $u_1$.
Since $K\subseteq R^2_d(G,s)$ so $\exit{G}{K}{s}{v} \in R^2_d(G,s)$. Since $\dist{G}{s}{v} \leq d$ so $\dist{G}{s}{u_4} \leq d$ and $u_4 \in R^2_n(G,s) \cap (R^1_d(G,s) \cup R^2_d(G,s))$. Hence, by construction of the seeking strategy, $u_4$ is visited before $u_1$ if
and only if either $u_2$ or $u_3$ is visited before $u_1$. This happens with probability $2/3$ and so $\Ex(X_{vt}) = 2/3$ in this case.

Suppose that $K \setminus R^2_d(G,s) \neq \varnothing$. In this case at least one node in the cycle is not reachable by two paths of length at most $d$.
Since $K \cap R^2_d(G,s) \neq \varnothing$ so every node in the cycle is reachable by at least one path of length at most $d$ from $s$.
Thus the set of visited nodes contains all the nodes of the cycle if and only if both $u_2$ and $u_3$ are visited before $u_1$.
If $v \in R^2_d(G,s)$ then $v$ is visited before $t$ if either $u_2$ or $u_3$ is visited before $t$. Hence $\Ex(X_{vt}) = 2/3$ in this case.
If $v \notin R^2_d(G,s)$ then $v \in R^1_d(G,s)$ and node $v$ is visited before $t$ if either $u_2$ is visited before $u_1$ and $u_3$, or 
$u_3$ is visited before $u_2$ and $u_2$ is visited before $u_1$. In the former case $v$ is visited before cycle $K$ is discovered and the case happens with probability $1/3$. In the latter case the cycle is discovered before $u_1$ is visited and node $v$ is visited before $u_1$, because after discovering the cycle
$u_1$ is visited after all the nodes reachable by two paths from $s$, at least one of them of the length at most $d$,  are visited. The case happens with probability $1/6$. Thus $v$ is visited before $t$ with probability $1/2$ and so $\Ex(X_{vt}) = 1/2$.

Suppose that $t \in R^2_n(G,s)$, i.e. there are two paths from $s$ to $t$ in $G$. 

Suppose that $t \in K$ and $v \in K$, i.e. both $t$ and $v$ belong to the cycle. Let $u_1 \in K$ be the successor of $\entrance{G}{K}{s}$ in the cycle that is closer to $t$ and $u_2 \in K$ be the successor of $\entrance{G}{K}{s}$ in the cycle that is closer two $v$. Since $\dist{G}{s}{t} \leq d$ and $\dist{G}{s}{v} \leq d$ so $v$ is visited before $t$ if and only if $u_2$ is visited before $u_1$. This happens with probability $1/2$ and so $\Ex(X_{vt}) = 1/2$ in this case.

Suppose that $t \in R^2_n(G,s)\setminus K$, i.e. there are two paths from $s$ to $t$ in $G$ and $t$ is not in the cycle. Let $u_1\in K$ be the successor of $\entrance{G}{K}{s}$ in the cycle that is closer to $t$ than the other successor of $\entrance{G}{K}{s}$ in the cycle, $u_2\in K$.
Let $u_3$ be the successor of $\exit{G}{K}{s}{t}$ on the path from $s$ to $t$. 

\begin{figure}[h]
\begin{center}
  \includegraphics{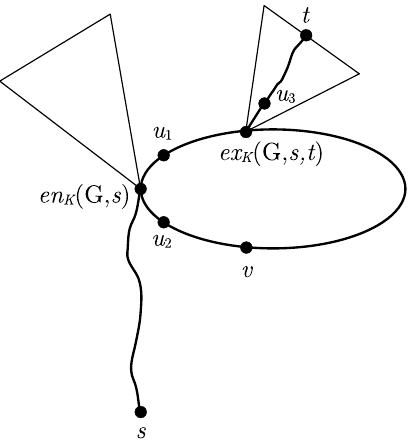}
\end{center}
\caption{Paths from $s$ to $t$ and from $s$ to $v$ in the case of $t \in R^1_n(G,s)\setminus K$ and $v \in K$.}
  \label{fig:conf2}
\end{figure}

Suppose that $v \in K$ (c.f. Figure~\ref{fig:conf2}). Since $v$ is not on every path from $s$ to $t$ of length at most $d$ so there exists a path from $s$ to $t$ of length at most $d$ that does not contain $v$ and, consequently, contains $u_1$. Suppose that $v \in R^2_d(G,s)$. In this case $v$ is visited before $t$ if and only if either $u_2$ is visited before $u_1$ or $u_1$ is visited before $u_2$ and $u_2$ is visited before $u_3$. This happens with probability $1/2 + 1/2 \cdot 1/2 = 3/4$. Hence $\Ex(X_{vt}) = 3/4$ in this case.

Suppose that $v \notin R^2_d(G,s)$, in which case $v \in R^1_d(G,s)$. If the path from $s$ to $v$ of length at most
$d$ contains $u_1$ then $v$ is visited before $t$ with probability $1/2$. Otherwise the path contains $u_2$ and $v$ is visited before $t$ if and only if $u_2$ is visited before $u_1$. This happens with probability $1/2$. 
Hence $\Ex(X_{vt}) = 1/2$ in this case.

\begin{figure}[h]
\begin{center}
  \includegraphics{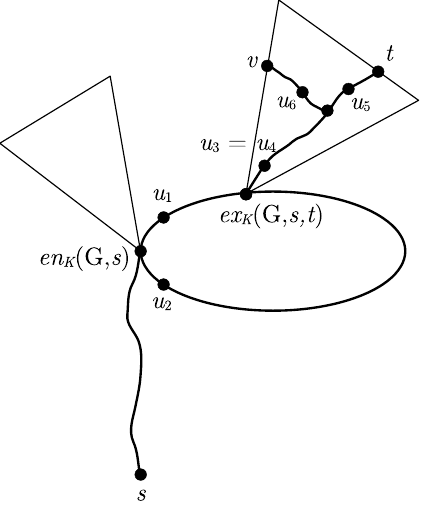}
\end{center}
\caption{Paths from $s$ to $t$ and from $s$ to $v$ in the case of $t \in R^2_n(G,s)\setminus K$, $v \in R^2_n(G,s)\setminus K$, and $u_3 = u_4$.}
  \label{fig:conf3}
\end{figure}

Suppose that $v \in R^2_n(G,s) \setminus K$, i.e. there are two paths from $s$ to $v$ in $G$ and $v$ is not in the cycle. Let $u_4$ be the successor of
$\exit{G}{K}{s}{v}$ on the path from $s$ to $v$. 
Consider first the case of $u_3 = u_4$. Let $u_5$ be the first node on the path from $u_3$ to $t$ that is not on the path from $u_3$ to $v$ and let $u_6$ be the first node on the path from $u_3$ to $v$ that is not on the path from $u_3$ to $t$ (c.f. Figure~\ref{fig:conf3}). If $t \in R^2_d(G,s)$ and $v \in R^2_d(G,s)$ then
$v$ is visited before $t$ if and only if $u_6$ is visited before $u_5$. This happens with probability $1/2$ and so $\Ex(X_{vt}) = 1/2$. 
Suppose that $t \in R^1_d(G,s)$. In this case there is a unique path of length at most $d$ from $s$ to $t$ and, since $u_1$ is closer to $t$ than $u_2$, 
this path contains $u_1$. Suppose that $v \in R^2_d(G,s)$. If the cycle is discovered (i.e. both $u_1$ and $u_2$ are visited before $u_3$, which happens with probability $1/2 \cdot 1/2 + 1/2 \cdot 1/2 = 1/2$) then $v$ is visited before $t$ if and only if $u_6$ is visited before $u_5$. Similarly, if $u_1$ is visited before $u_3$ and $u_3$ is visited before $u_2$ (which happens with probability $1/2\cdot 1/2 = 1/4$) then $v$ is visited before $t$ if and only if $u_6$ is visited before $u_5$. If $u_2$ is visited before $u_3$ and $u_3$ is visited before $u_1$ then $v$ is visited before $t$ with probability $1$ because in this case $t$ cannot be reached by the seeking strategy until $u_1$ is visited, i.e. the cycle is discovered, and $v$ is visited before $u_1$ if $u_3$ is visited before $u_1$. 
Hence $v$ is visited before $t$ with probability $1/2\cdot 1/2 + 1/4 \cdot 1/2 + 1/4 \cdot 1 = 5/8$ and so $\Ex(X_{vt}) = 5/8$.
Suppose that $v \in R^1_d(G,s)$. Since $u_1$ is closer to $t$ than $u_2$ so $u_1$ is closer to $u_3$ than $u_2$ and, consequently, $u_1$ is closer to $v$ than $u_2$. Thus $v$ and $t$ are visited only after $u_1$ is visited and $v$ is visited before $t$ if and only if $u_6$ if visited before $u_5$.
This happens with probability $1/2$ and so $\Ex(X_{vt}) = 1/2$.
 
\begin{figure}[h]
\begin{center}
  \includegraphics{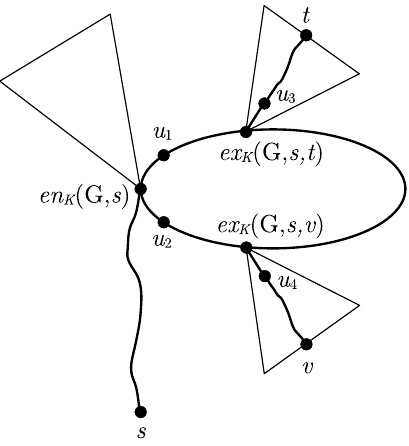}
\end{center}
\caption{Paths from $s$ to $t$ and from $s$ to $v$ in the case of $t \in R^2_n(G,s)\setminus K$, $v \in R^2_n(G,s)\setminus K$, and $u_3 \neq u_4$.}
  \label{fig:conf4}
\end{figure}

Second, consider the case of $u_3 \neq u_4$ (c.f. Figure~\ref{fig:conf4}).

Suppose that $K \subseteq R^2_d(G,s)$, i.e. all the nodes in the cycle are reachable from $s$ by two paths of length at most $d$. 

Suppose that $\{\exit{G}{K}{s}{t},\exit{G}{K}{s}{v}\} \subseteq R^2_{d-1}(G,s)$, so that $u_3 \in R^2_d(G,s)$ and $u_4 \in R^2_d(G,s)$,
Consider the following three subcases: 
\begin{enumerate*}[label=(\alph*)]
\item $t \in R^2_d(G,s)$ and $v \in R^2_d(G,s)$,\label{p:case:1}
\item $t \in R^1_d(G,s)$ and $v \in R^1_d(G,s)$,\label{p:case:2}
\item $t \in R^1_d(G,s)$ and $v \in R^2_d(G,s)$.\label{p:case:3}
\end{enumerate*}
In case~\ref{p:case:1}, $t$ is visited before $v$ if and only if $u_4$ is visited before $u_3$, which happens with probability $1/2$ in both cases of $u_1$ being visited before $u_2$ and $u_2$ being visited before $u_1$. Hence $v$ is visited before $t$ with probability $1/2$ and so $\Ex(X_{vt}) = 1/2$.

In case~\ref{p:case:2}, if $v$ is closer to $u_1$ than to $u_2$ then the unique path of length at most $d$ from $s$ to $v$ includes $u_1$ and does not include $u_2$. The same is the case for $t$, because $t$ is closer to $u_1$ than to $u_2$. Thus if $v$ is closer to $u_1$ than to $u_2$ then $v$ is visited before $t$ if and only if either $u_1$ is visited before $u_2$ and $u_4$ is visited before $u_3$ or $u_2$ is visited before $u_1$ and $u_4$ is visited before $u_3$. This happens with probability $1/2\cdot 1/2 + 1/2 \cdot 1/2 = 1/2$ and so $\Ex(X_{vt}) = 1/2$.
Notice that in the case of $u_2$ being visited before $u_1$, if $u_3$ was visited before $u_4$ and $u_1$ was visited then, although $t$ would not be visited before $u_1$ (because the path from $s$ to $t$ going through $u_2$ is of length greater than $d$), it would be visited before $v$ because the seeking strategy visits nodes in the ``first seen first visited'' order (i.e. those whose neighbour in the set of visited nodes is earlier in the sequence of visited nodes are visited before those whose neighbour in the set of visited nodes is later in the sequence). If, on the other hand, $u_4$ was visited before $u_3$ then, although $v$ would not be visited before $u_1$, it would be visited before $t$, again, because the seeking strategy visits nodes in the ``first seen first visited'' order.
If $v$ is closer to $u_2$ than to $u_1$ then the unique path of length at most $d$ from $s$ to $v$ includes $u_2$ and does not include $u_1$. On the other hand,
the unique path of length at most $d$ from $s$ to $t$ includes $u_1$ and does not include $u_2$ because $t$ is closer to $u_1$ than to $u_2$.
Thus in this case $\exit{G}{K}{s}{v}$ is closer to $u_2$ than $\exit{G}{K}{s}{t}$. Hence $v$ is visited before $t$ if either $u_1$ is visited before $u_2$
and $\exit{G}{K}{s}{v}$ is visited before $u_3$ or $u_2$ is visited before $u_1$ and $u_4$ is visited before $\exit{G}{K}{s}{t}$. This happens with probability 
$1/2\cdot 1/2 + 1/2\cdot 1/2 = 1/2$ and so $\Ex(X_{tv}) = 1/2$.
 
In case~\ref{p:case:3}, if $\exit{G}{K}{s}{v}$ is on the unique path of length at most $d$ from $s$ to $t$ then, since the path includes $u_1$ and does not include
$u_2$, $\exit{G}{K}{s}{v}$ is closer to $u_1$ than $\exit{G}{K}{s}{t}$. In this case $v$ is visited before $t$ if and only if either $u_1$ is visited before $u_2$ and $u_4$ is visited before $\exit{G}{K}{s}{t}$ or $u_2$ is visited before $u_1$ and either $\exit{G}{K}{s}{v}$ is visited before $u_3$ or $u_3$ is visited before $\exit{G}{K}{s}{t}$ and $u_4$ is visited before $u_1$. This happens with probability $1/2\cdot 1/2 + 1/2 \cdot (1/2 + 1/2 \cdot 1/2) = 5/8$ and so $\Ex(X_{tv}) = 5/8$.
If $\exit{G}{K}{s}{v}$ is not on the unique path of length at most $d$ from $s$ to $t$ then $\exit{G}{K}{s}{v}$ is closer to $u_2$ than $\exit{G}{K}{s}{t}$.
In this case $v$ is visited before $t$ if and only if either $u_1$ is visited before $u_2$ and $\exit{G}{K}{s}{v}$ is visited before $u_3$ or $u_2$ is visited before $u_1$ and $u_4$ is visited before $\exit{G}{K}{s}{t}$. This happens with probability $1/2\cdot 1/2 + 1/2 \cdot 1/2 = 1/2$ and so $\Ex(X_{tv}) = 1/2$.

Suppose that $K \setminus R^2_d(G,s) \neq \varnothing$, i.e. there exists nodes in the cycle which are reachable from $s$ by two paths of length at most $d$.
Suppose that $t \in R^2_d(G,s)$ and $v \in R^2_d(G,s)$. In this case both nodes $u_3$ and $u_4$ are visited if either $u_1$ or $u_2$ is visited but before all
the nodes in the cycle are visited. Therefore $v$ is visited before $t$ if and only if $u_4$ is visited before $u_3$. This happens with probability $1/2$ in the case when $u_1$ is visited before $u_2$ as well as in the case when $u_2$ is visited before $u_1$. Hence $v$ is visited before $t$ with probability $1/2$
and so $\Ex(X_{tv}) = 1/2$. Suppose that $t \in R^1_d(G,s)$ and $v \in R^2_d(G,s)$. Since $u_1$ is closer to $t$ than $u_2$ so $t$ is visited only after $u_1$ is
visited. Hence $v$ is visited before $t$ if and only if either $u_2$ is visited before $u_1$ or $u_1$ is visited before $u_2$ and $u_4$ is visited before $u_3$.
This happens with probability $1/2 + 1/2\cdot 1/2 = 3/4$. Therefore $\Ex(X_{tv}) = 3/4$ in this case.
Suppose that $t \in R^1_d(G,s)$ and $v \in R^1_d(G,s)$. If $u_1$ is closer to $v$ than $u_2$ then, since $u_1$ is closer to $t$ than $u_2$, $v$ is visited before $t$ if and only if $u_4$ is visited before $u_3$, which happens with probability $1/2$. If $u_2$ is closer to $v$ than $u_1$ then $v$ is visited before $t$
if and only if $u_2$ is visited before $u_1$, which happens with probability $1/2$. Hence, in either case, $\Ex(X_{tv}) = 1/2$.
\end{proof}

Since the DFS strategy is equal to the bounded DFS strategy with bound $n$, $\sigma^{\dfs} = \sigma^{\dfs}_n$, we have the following corollary from Lemma~\ref{lemma:nosteps:dfsd}.

\begin{corollary}
\label{cor:nosteps:dfs}
Let $G = \langle V, E\rangle$ be a graph containing at most one cycle $K$ and let $t \in V$ be either a leaf or a node in $K$. Under the DFS seeking strategy $\sigma^{\dfs}$, for any node $v \in V\setminus P_G(s,t)$ such that $t\notin P_G(s,v)$,
\begin{itemize}
\item if $t \in R^{1}_n(G,s) \setminus R^1_n(G,s,\entrance{G}{K}{s})$ then $\Ex(X_{vt}) = 1/2$,
\item if $t \in R^1_n(G,s,\entrance{G}{K}{s})$ then
\begin{equation*}
\Ex(X_{vt}) = \begin{cases}
              \frac{1}{2}, & \textrm{if $v\in R^1_n(G,s)$,}\\
              \frac{2}{3}, & \textrm{if $v \in R^2_n(G,s)$,}
              \end{cases}
\end{equation*}
\item if $t \in K$ and $v \in K$ then $\Ex(X_{vt}) = 1/2$,
\item if $t \in R^2_n(G,s) \setminus K$ and $v \in K$ then $\Ex(X_{vt}) = 3/4$,
\item if $\{t,v\} \subseteq R^2_n(G,s) \setminus K$ then $\Ex(X_{vt}) = 1/2$.
\end{itemize}
\end{corollary}

Lastly, the following lemma characterizes the expected value of $X_{vt}$ for the adjusted DFS seeking strategy.
\begin{lemma}
\label{lemma:nosteps:adfs}
Let $d > 0$, $G = \langle V, E\rangle$ be a graph containing at most one cycle $K$, and let $t \in V$ be either a leaf or a node in $K$. Under adjusted DFS seeking strategy $\sigma^{\adjdfs}$, for any node $v \in V\setminus P_G(s,t)$ such that $t\notin P_G(s,v)$,
\begin{itemize}
\item if $t \in R^{1}_n(G,s) \setminus R^1_n(G,s,\entrance{G}{K}{s})$ then $\Ex(X_{vt}) = 1/2$,
\item if $t \in R^1_n(G,s,\entrance{G}{K}{s})$ then
\begin{equation*}
\Ex(X_{vt}) = \begin{cases}
              \frac{1}{2}, & \textrm{if $v\in R^1_n(G,s)$,}\\
              \frac{2}{3}, & \textrm{if $v \in K$,}\\
              \frac{1}{3}, & \textrm{if $v \in R^2_n(G,s)\setminus K$,}
              \end{cases}
\end{equation*}
\item if $t \in K$ and $v \in K$ then $\Ex(X_{vt}) = 1/2$,
\item if $t \in R^2_n(G,s) \setminus K$ and $v \in K$ then $\Ex(X_{vt}) = 3/4$,
\item if $\{t,v\} \subseteq R^2_n(G,s) \setminus K$ then $\Ex(X_{vt}) = 1/2$.
\end{itemize}
\end{lemma}

\begin{proof}
Take any connected graph $G$ with starting node $s \in V$ such that either $\cycles(G) = \varnothing$ or $\cycles(G) = \{K\}$, so that $K$ is the unique cycle in $G$.
Take any two nodes $t,v \in V(G)$ such that $v \neq t$ and $t$ is either a leaf or a node in $K$.
Suppose that $v \notin P_G(s,t)$ and $t \notin P_G(s,v)$.

Suppose that $t \in R^1_n(G,s) \setminus R^1_n(G,s,\entrance{G}{K}{s})$, i.e. there is a unique path from $s$ to $t$ in $G$ and that path does not contain any node of the cycle. Since neither $v \in V(G) \setminus P(s,t)$ nor $t \notin P(s,v)$ so $\Ex(X_{vt}) = 1/2$. 
This follows by the same arguments as those used in proof of Lemma~\ref{lemma:nosteps}.

Suppose that $t \in R^1_n(G,s,\entrance{G}{K}{s})$, i.e. there is a unique path from $s$ to $t$ in $G$ and that path contains a node of the cycle. In this case
the path from $s$ to $t$ contains $\entrance{G}{K}{s}$ and no other node of the cycle. 
Suppose that $v \in R^1_n(G,s)$. Then there is exactly one path from $s$ to $v$ in $G$ and $\Ex(X_{vt}) = 1/2$, by analogous arguments to those used in proof of Lemma~\ref{lemma:nosteps}. 

Suppose that $v \in R^2_n(G,s)$. Let $u_1$ be the successor of $\entrance{G}{K}{s}$ on the path from $s$ to $t$, and $u_2$ and $u_3$ be the successors of $\entrance{G}{K}{s}$ on the paths from $s$ to $v$. In the case of $v \notin K$, let $u_4$ be the successor $\exit{G}{K}{s}{v}$ on the path from $s$ to $v$ (c.f. Figure~\ref{fig:conf1}). 

If $v \in K$ then $v$ is visited before $t$ if and only if either $u_2$ and $u_3$ is visited before $u_1$. This happens
with probability $2/3$ and so $\Ex(X_{vt}) = 2/3$. Suppose that $v \notin K$. In this case $v$ is visited before $t$ if and only if $u_4$ is visited before $u_1$. This happens if and only if either
$u_2$ is visited before $u_1$ and $u_3$ and then $u_4$ is visited before $u_3$ or $u_3$ is visited before $u_1$ and $u_2$ and then $u_4$ is visited before $u_2$. To see why these conditions are necessary, notice that if $u_2$ is visited before $u_1$ and then $u_3$ is visited before $u_4$ then the cycle is discovered before $u_4$ is visited and, by construction of the seeking strategy, $u_1$ is visited before $u_4$. Consequently, $t$ is visited before $u_4$ and $v$. Similar argument applies to the case of $u_3$ being visited before $u_1$ and then $u_2$ being visited before $u_4$. Thus $v$ is visited before $t$ with probability $1/3 \cdot 1/2 + 1/3 \cdot 1/2 = 1/3$. 

Suppose that $t \in R^2_n(G,s)$, i.e. there are two paths from $s$ to $t$ in $G$.
Let $u_1 \in K$ be the successor of $\entrance{G}{K}{s}$ in the cycle that is closer to $t$ and $u_2 \in K$ be the successor of $\entrance{G}{K}{s}$ in the cycle that is closer two $v$. 

Suppose that $t \in K$ and $v \in K$, i.e. both $t$ and $v$ belong to the cycle.  Then $v$ is visited before $t$ if and only if $u_2$ is visited before $u_1$. This happens with probability $1/2$ and so $\Ex(X_{vt}) = 1/2$ in this case.

Suppose that $t \in R^2_n(G,s)\setminus K$, i.e. there are two paths from $s$ to $t$ in $G$ and $t$ is not in the cycle. Let $u_3$ be the successor of $\exit{G}{K}{s}{t}$ on the path from $s$ to $t$. 
Suppose that $v \in K$ (c.f. Figure~\ref{fig:conf2}). Then $v$ is visited before $t$ if and only if either $u_2$ is visited before $u_1$ or $u_1$ is visited before $u_2$ and $u_2$ is visited before $u_3$. This happens with probability $1/2 + 1/2 \cdot 1/2 = 3/4$. Hence $\Ex(X_{tv}) = 3/4$ in this case.
Suppose that $v \in R^2_n(G,s) \setminus K$, i.e. there are two paths from $s$ to $v$ in $G$ and $v$ is not in the cycle. Let $u_4$ be the successor of $\exit{G}{K}{s}{v}$ on the path from $s$ to $v$. 
Consider first the case of $u_3 = u_4$. Let $u_5$ be the first node on the path from $u_3$ to $t$ that is not on the path from $u_3$ to $v$ and let $u_6$ be the first node on the path from $u_3$ to $v$ that is not on the path from $u_3$ to $t$ (c.f. Figure~\ref{fig:conf3}). Then $v$ is visited before $t$ if and only if $u_6$ is visited before $u_5$. This happens with probability $1/2$ and so $\Ex(X_{vt}) = 1/2$. 
Second, consider the case of $u_3 \neq u_4$ (c.f. Figure~\ref{fig:conf4}).
In this case $t$ is visited before $v$ if and only if $u_4$ is visited before $u_3$, which happens with probability $1/2$ in both cases of $u_1$ being visited before $u_2$ and $u_2$ being visited before $u_1$. Hence $v$ is visited before $t$ with probability $1/2$ and so $\Ex(X_{tv}) = 1/2$.
\end{proof}

The following corollary from Lemmata~\ref{lemma:nosteps:dfsd} and~\ref{lemma:nosteps:adfs} and Corollary~\ref{cor:nosteps:dfs} characterizes the probabilities of a node being visited before another node for the seeking strategies being a convex combination of strategies $\dfs_d$, $\adjdfs$, and bounded $\dfs$.

\begin{corollary}
\label{cor:nosteps:star}
Let $d > 0$, $G = \langle V, E\rangle$ be a graph containing at most one cycle $K$, and let $t \in V$ be either a leaf or a node in $K$. Under the seeking strategy 
\begin{equation*}
\sigma_d^{\star} = \frac{3}{8} \sigma^{\adjdfs} + \frac{3}{8}\sigma^{\dfs} + \frac{1}{4} \sigma^{\dfs_d},
\end{equation*}
with $\lambda \in [0,1]$, for any node $v \in V\setminus P_G(s,t)$ such that $t\notin P_G(s,v)$,
\begin{itemize}
\item if $t \in R^{1}_n(G,s) \setminus R^1_n(G,s,\entrance{G}{K}{s})$ then
\begin{equation*}
\Ex(X_{vt}) = \begin{cases}
              \frac{1}{2}, & \textrm{if $\dist{G}{s}{v} \leq d$,}\\
              \frac{3}{8}, & \textrm{if $\dist{G}{s}{v} > d$.}
              \end{cases}
\end{equation*}
\item if $t \in R^1_n(G,s,\entrance{G}{K}{s})$ then
\begin{equation*}
\Ex(X_{vt}) = \begin{cases}
              \frac{1}{2}, & \textrm{if $v\in R^1_d(G,s) \cap R^1_n(G,s)$,}\\
              \frac{13}{24}, & \textrm{if $v\in R^1_d(G,s) \cap R^2_n(G,s)$ and $K\subseteq R^2_d(G,s)$,}\\
              \frac{1}{2}, & \textrm{if $v\in R^1_d(G,s) \cap (R^2_n(G,s)\setminus K)$ and $K \nsubseteq R^2_d(G,s)$,}\\
              \frac{5}{8}, & \textrm{if $v\in R^1_d(G,s) \cap K$ and $K \nsubseteq R^2_d(G,s)$,}\\
              \frac{2}{3}, & \textrm{if $v\in R^2_d(G,s) \cap K$,}\\
              \frac{13}{24}, & \textrm{if $v\in R^2_d(G,s)\setminus K$,}\\
              \frac{3}{8}, & \textrm{if $v\in R^1_n(G,s) \setminus R^1_d(G,s)$,}\\
              \frac{1}{2}, & \textrm{if $v\in K \setminus (R^1_d(G,s) \cup R^2_d(G,s))$,}\\
              \frac{3}{8}, & \textrm{if $v\in R^2_n(G,s) \setminus (R^1_d(G,s) \cup R^2_d(G,s)\cup K)$,}\\
              \end{cases}
\end{equation*}
\item if $t \in K$ and $v \in K$ then $\Ex(X_{vt}) = 1/2$,
\item if $t \in R^2_n(G,s) \setminus K$ and $v \in K$ then
\begin{equation*}
\Ex(X_{vt}) = \begin{cases}
              \frac{13}{16}, & \parbox[t]{9cm}{if $t \in R^1_d(G,s) \cap R^1_d(G,s,v)$,}\\
              \frac{3}{4}, & \parbox[t]{8cm}{if $t \notin R^1_d(G,s) \cap R^1_d(G,s,v)$ and $v \in R^2_d(G,s)$,}\\
              \frac{11}{16}, & \parbox[t]{8cm}{if $t \notin R^1_d(G,s) \cap R^1_d(G,s,v)$ and $v \in R^1_d(G,s)$.}\\
              \frac{9}{16}, & \parbox[t]{10cm}{if $t \notin R^1_d(G,s) \cap R^1_d(G,s,v)$ and $v \notin R^1_d(G,s)\cup R^2_d(G,s)$.}\\
              \end{cases}
\end{equation*}

\item if $\{t,v\} \subseteq R^2_n(G,s) \setminus K$ and $K\subseteq R^2_d(G,s)$ then
\begin{align*}
\Ex(X_{vt}) & =  \frac{1}{2} + \frac{1}{32}[v \in R^2_d(G,s)][t \in R^1_d(G,s) \cap R^1_d(G,s,\exit{G}{K}{s}{v})]
\end{align*}

\item if $\{t,v\} \subseteq R^2_n(G,s) \setminus K$ and $K\setminus R^2_d(G,s) \neq \varnothing$ then
\begin{align*}
\Ex(X_{vt}) =  \frac{1}{2} & + \frac{1}{16}[t \in R^1_d(G,s)][v \in R^2_d(G,s)]
\end{align*}
\end{itemize}
\end{corollary}

\begin{proof}[Proof of Proposition~\ref{pr:cycle:lowerbound}]
To obtain the upper bound on the number nodes visited by strategy $\sigma_d^{\star}$ until a node at distance $d$ from node $s$ is reached notice that for any node on all the paths from $s$ to $t$, $\Ex(X_{vt}) = 1$. Moreover,
since, for any two nodes $u$ and $v$, $X_{uv} = 1 - X_{vu}$ so $\Ex(X_{uv}) = 1-\Ex(X_{vu})$. By these two observations, together with Corollary~\ref{cor:nosteps:star}, for any node $t$ at distance $d$ from $s$ and any node $v$, 
\begin{equation*}
\Ex(X_{vt}) = \begin{cases}
                 1, & \textrm{if $v \in P_G(s,t)$} \\
                 \frac{5}{8}, & \textrm{if $t \in R^1_n(G,s,\entrance{G}{K}{s})$ and $v\in R^1_d(G,s) \cap K$ and $K \nsubseteq R^2_d(G,s)$,}\\
                 \frac{2}{3}, & \textrm{if $t \in R^1_n(G,s,\entrance{G}{K}{s})$ and $v\in R^2_d(G,s) \cap K$,}\\
                 \frac{13}{16}, & \textrm{if $t \in R^2_n(G,s) \setminus K$ and $t \in R^1_d(G,s) \cap R^1_d(G,s,v)$,}\\
                 \frac{3}{4}, & \textrm{if $t \in R^2_n(G,s) \setminus K$, $t \notin R^1_d(G,s) \cap R^1_d(G,s,v)$ and $v \in R^2_d(G,s)$,}\\
                 \frac{11}{16}, & \textrm{if $t \in R^2_n(G,s) \setminus K$, $t \notin R^1_d(G,s) \cap R^1_d(G,s,v)$ and $v \in R^1_d(G,s)$.}\\
                 \end{cases}
\end{equation*}
and in all the remaining cases, $\Ex(X_{vt}) \leq 9/16$.
Notice that the number of nodes $v$ satisfying the first six cases is at most $d$. In the first case,
there are at most $d$ nodes that are on every path from $s$ to $t$. 
The second and the sixth case is satisfied only for nodes $v$ that are in the cycle and are at distance at most $d$ from the source. There are at most $2d-1$ such nodes. The third and the fifth case is satisfied only for nodes $v$ that are in the cycle and are reachable by two paths of length at most $d$ from the source. There are at most $d$ such nodes. The fourth case is satisfied only for nodes $v$ that are on the unique path of length at most $d$ from $s$ to $t$ (the case applies when $t$ is reachable by two paths from $s$ but only one of them is of length at most $d$). There are at most $d-1$ such nodes. 

Hence the expected number of nodes visited by seeking strategy $\sigma_d^{\star}$ before $t$ is reached is not greater than
\begin{equation*}
\frac{9}{16}(n - d) + \frac{11}{16} (2d-1) = \frac{9}{16}n + \frac{13d - 11}{16}.
\end{equation*}
\end{proof}

\end{document}